\documentclass{amsart}
\usepackage{amssymb}
\usepackage{amsfonts}

\newtheorem{theorem}{Theorem}
\theoremstyle{plain}

\newtheorem{axiom}{Axiom}

\newtheorem{definition}{Definition}

\newtheorem{lemma}{Lemma}

\newtheorem{proposition}{Proposition}

\numberwithin{equation}{section}
\input{tcilatex}

\begin{document}
\title[Complex bodies with memory effects]{\textbf{Complex bodies with
memory: linearized setting}}
\author{Paolo Maria Mariano}
\address{DICeA, University of Florence, via Santa Marta 3, I-50139 Firenze
(Italy)}
\email{paolo.mariano@unifi.it}
\author{Paolo Paoletti}
\address{DSI, University of Florence, via Santa Marta 3, I-50139 Firenze
(Italy)}
\email{paolo.paoletti@dsi.unifi.it}
\date{May 16th, 2008}

\begin{abstract}
The mechanics of complex bodies with memory effects is discussed in
linearized setting. The attention is focused on the characterization of free
energies in terms of minimum work and maximum recoverable work in the bulk
and along a discontinuity surface endowed with its own surface energy, a
surface internal to the body. To this aim, use is made of techniques
proposed by Del Piero. Consequences of the Clausius-Duhem inequality are
investigated for complex bodies with instantaneous linear elastic response.
\end{abstract}

\maketitle

\section{Introduction}

Some materials display sensibility to the past history of their present
state: they are called the materials with memory. A paradigmatic (although
special) example of description of hereditary behavior is the standard
viscoelasticity, even in small strain regime, for example the one described
by the rheological models by Newton, Maxwell or Kelvin. The pioneering work
by Volterra \cite{Vol1, Vol2} opened the way to the analysis of complicated
constitutive structures built up on the histories of state variables
(basically their graphs in time).

The topic received careful attention in the late 1950s, the subsequent
decade and more. Researches were developed with the aim of establishing a
theory of linear and non-linear behavior of materials with memory,
especially in the case of fading memory \cite{KM, CN, CNbis, Bernardo,
Beenardo2}. All these works generated specific developments and general
effort toward ordering of the existing results in a clear and systematic
way. The articles on this topic are manifold and it is difficult to list all
of them. Essential examples are \cite{Fichera, CoMi, CoMi1, CO, GuHe,
LeiFis, D, Daf2, Daf3, Daf4, Day1, Day2, Day3, Ant, DP, DPD, FGG, FGM, FM,
McT, RHN}, all dealing with the description of memory effects in Cauchy
bodies that are those bodies the morphology of every material element of
which is described by the sole place that its centre of mass occupies in
space.

Essential questions have been tackled in the representation of the
mechanical behavior of the materials with memory. Some of them have
foundational nature: (\emph{i}) the meaning to be given to the notion of
state, (\emph{ii}) the definition of appropriate free energies, (\emph{iii})
the nature of the chain rule when functionals of histories are called upon, (%
\emph{iv}) the correct use of Clausius-Duhem inequality to find a priori
constitutive restrictions.

In particular, in linear viscoelasticity items (\emph{i}) and (\emph{ii})
have been successfully tackled in \cite{DP} and \cite{DPD}: viscoelastic
material elements of Cauchy bodies have been described as thermodynamic
systems in the sense of \cite{CO1}.

Here, techniques proposed by Del Piero in \cite{DP} are used extensively to
address the mathematical modelling of the behavior of \emph{complex bodies
with memory} in the general model-building framework of the mechanics of
complex bodies (different analyses on the specific case of micromorphic
materials with memory are available in \cite{Erin, IesSca}).

Bodies are called complex when changes in the molecular or crystalline
texture at various microscopic scales (substructure) influence the
macroscopic behavior through peculiar actions. Examples are manifold: liquid
crystals, ferroelectrics, quasicrystals, nematic elastomers,
magnetostrictive solids etc. Although all these examples are referred to a
variety of phenomena, common essential features can be referred to a unique
abstract model-building framework for the mechanics of complex bodies (see 
\cite{C89, M02, M08}). Such a framework unifies in a single format the
existing models of special classes of complex bodies and is a flexible tool
for analyzing new materials.

Some complex bodies, as for example relaxor ferroelectrics \cite{Coh, VuR},
exhibit memory of substructural events. Such a circumstance motivates the
analysis of the mechanics and thermodynamics of bodies with memory. The
setting is the linear one for the sake of simplicity. The representation of
the substructural morphology of bodies is maintained abstract in order to
include a number of special cases as large as possible.

Free energies are characterized in terms of minimum work and maximum
recoverable work in the bulk and along a discontinuity surface endowed with
its own energy. Consequences of the Clausius-Duhem inequality are
investigated for complex bodies with instantaneous linear elastic response.

\section{Kinematics of complex bodies}

\subsection{Generalities}

A regular region $\mathcal{B}$ in the ambient space $\mathbb{R}^{3}$ -
regular in the sense that it is a `fit region' or, more simply, an open set
with Lipschitz boundary - is selected to host a body in its macroscopic
reference configuration. Every point $x$ from $\mathcal{B}$ represents a
material element. It is assumed that subsequent configurations are achieved
by means of differentiable bijections (transplacements) from $\mathcal{B}$
to a copy $\mathbb{\hat{R}}^{3}$\ of the ambient space, obtained by means of
an isomorphism $i:\mathbb{R}^{3}\rightarrow \mathbb{\hat{R}}^{3}$. Maps $%
x\longmapsto y:=y\left( x\right) \in \mathbb{\hat{R}}^{3}$, $x\in \mathcal{B}
$, are then defined. Their spatial derivative is indicated by $F:=Dy\left(
x\right) \in Hom\left( T_{x}\mathcal{B},T_{y\left( x\right) }\mathcal{B}%
_{a}\right) $, with $\mathcal{B}_{a}:=y\left( \mathcal{B}\right) $, and is
such that $\det F>0$. Comparison between the metric $g$ in the actual shape
in $\mathcal{B}_{a}$ and the natural metric $\gamma $ in $\mathcal{B}$
allows one to measure crowding and shearing of material elements. The tensor 
$E:=\frac{1}{2}\left( y^{\#}g-\gamma \right) $, where $y^{\#}g$ is the
pull-back of $g$ through $y$\ given in components by $\left( y^{\#}g\right)
_{RS}=F_{R}^{i}g_{ij}F_{S}^{j}$, is then a true measure of deformation: it
vanishes under global rigid transplacements. Motions are then time
parametrized families of transplacements:%
\begin{equation*}
\left( x,t\right) \longmapsto y:=y\left( x,t\right) \in \mathbb{\hat{R}}^{3},%
\text{ \ \ }x\in \mathcal{B},\text{ }t\in \lbrack 0,d].
\end{equation*}%
Sufficient smoothness in time $t$ is presumed. The macroscopic velocity is
defined by $\dot{y}:=\frac{d}{dt}y\left( x,t\right) $ in the referential
description (that is as a field over the tube $\mathcal{B}\times \lbrack
0,d] $). The standard kinematics of deformable bodies is represented this
way. No geometrical information on the material texture at scales lower than
the macroscopic one (substructure) is commonly added in the description of
the morphology of the body under scrutiny.

When materials display acute sensibility to phenomena at minute scales, a
representation of such phenomena, combined with the description of the
macroscopic behavior is required. The standard kinematics is then enriched.
A descriptor of the material texture (a morphological descriptor of the
substructure) is assigned to each point. It is selected in a set $\mathcal{M}
$. A field%
\begin{equation*}
\left( x,t\right) \longmapsto \nu :=\nu \left( x,t\right) \in \mathcal{M},\
\ x\in \mathcal{B},\text{\ }t\in \lbrack 0,d],
\end{equation*}%
is then introduced: differentiability in space and sufficient smoothness in
time are presumed. The time rate of change of the morphological descriptor
field in referential description is $\dot{\nu}:=\frac{d}{dt}\nu \left(
x,t\right) \in T_{\nu \left( x,t\right) }\mathcal{M}$, the spatial
derivative is indicated by $N:=D\nu \left( x\right) \in Hom\left( T_{x}%
\mathcal{B},T_{\nu \left( x\right) }\mathcal{M}\right) $.

In selecting a specific morphological descriptor $\nu $ of the material
texture (a parameter also called a fabric tensor or an order parameter) one
choose the prominent geometrical features of the minute world inside the
generic material element to be described, transferring their peculiarities
at gross scale. However, the construction of the essential structures of the
mechanics of complex bodies requires only that $\mathcal{M}$\ be considered
as an abstract finite-dimensional differentiable manifold. Specific
geometrical property of $\mathcal{M}$ have often a clear physical meaning. A
metric is the basic ingredient for the representation of the independent
peculiar kinetic energy pertaining to the substructure, if such energy can
exist in special cases (see \cite{CG}). A connection allows one to represent
contact substructural interaction (microstresses) and to decompose in
invariant way them from the self-actions occurring in each material element
(see \cite{Se}). For these reasons, the specific nature of $\mathcal{M}$ is
left unspecified in the developments below. The subsequent results then hold
for a wide class of complex bodies.

\subsection{A discontinuity surface}

A surface%
\begin{equation*}
\Sigma :=\left\{ x\in cl\mathcal{B}\text{, \ }f\left( x\right) =0\right\} ,
\end{equation*}%
with $f$ a smooth function, is selected in $\mathcal{B}$. At $x$ the normal $%
m$ to $\Sigma $ is defined by%
\begin{equation*}
m=\frac{\nabla f\left( x\right) }{\left\vert \nabla f\left( x\right)
\right\vert },
\end{equation*}%
and orients $\Sigma $ locally. Notice that by such a definition the normal
is considered as a co-vector. The projector over $\Sigma $\ is the
second-rank tensor $\left( I-m\otimes m\right) $, with $I$ the identity. Let 
$x\longmapsto a\left( x\right) $ be a differentiable field over $\mathcal{B}$%
. Its surface gradient at $x\in \Sigma $ is given by $\nabla _{\Sigma
}a:=\nabla a\left( I-m\otimes m\right) $. The trace of $\nabla _{\Sigma }a$
defines the surface divergence of $a$ at $x$, namely $Div_{\Sigma
}a=tr\nabla _{\Sigma }a$.

Assume that $x\mapsto a\left( x\right) $\ takes values in a linear space. If
it is piecewise differentiable and suffers a bounded discontinuity over $%
\Sigma $, its jump $[a]$ across $\Sigma $ is defined by $[a]:=a^{+}-a^{-}$
at every $x\in \Sigma $, that is by the difference between the inner and the
outer traces of $a$\ at $\Sigma $, defined by the limits $a^{\pm
}:=\lim_{\varepsilon \rightarrow 0^{+}}a\left( x\pm \varepsilon m\right) $.
The average $\left\langle a\right\rangle $\ of $a$ across $\Sigma $ at every 
$x$ is defined by $2\left\langle a\right\rangle :=a^{+}+a^{-}$. For every
pair of fields $a_{1}$ and $a_{2}$ with the same properties of $a$, the
relation $[a_{1}a_{2}]=[a_{1}]\left\langle a_{2}\right\rangle +\left\langle
a_{1}\right\rangle [a_{2}]$ holds if the product $a_{1}a_{2}$\ is defined in
distributive way.

It is assumed here that both $x\longmapsto F$ and $x\longmapsto N$ are
discontinuous across $\Sigma $, while the field $x\longmapsto \nu $ is
continuous there. The symbols $\mathbb{F}$ and $\mathbb{N}$ denote the
surface gradients of deformation and the morphological descriptor,
respectively. They are defined by%
\begin{equation*}
\mathbb{F}:=\left\langle F\right\rangle \left( I-m\otimes m\right) ,\text{ \
\ }\mathbb{N}:=\left\langle N\right\rangle \left( I-m\otimes m\right) .
\end{equation*}

\subsection{Linearized kinematics}

Convenience suggests the introduction of the \emph{displacement field}%
\begin{equation*}
\left( x,t\right) \longmapsto u:=u\left( x,t\right) =y\left( x,t\right)
-i\left( x\right) ,\text{ \ \ }\left( x\mathbf{,}t\right) \in \mathcal{B}%
\times \left[ 0,d\right] .
\end{equation*}%
The spatial derivative of $x\longmapsto u$ is indicated here by $W:=Du\left(
x,t\right) $. One gets obviously $\dot{y}\left( x,t\right) =\dot{u}\left(
x,t\right) $ and $F=I+W$, with $I$ the second rank identity tensor. The
condition%
\begin{equation*}
\left\vert W\right\vert <<1
\end{equation*}%
defines the infinitesimal deformation regime. The deformation tensor $E$ is
then substituted by its linearized part $\varepsilon :=symW$. Moreover, no
distinction is also made between $\mathcal{B}$ and $\mathcal{B}_{a}$ in the
sense that $\dot{y}\approx u$ at every $x$ in $\mathcal{B}$. No distinction
is also made between $x\longmapsto \nu $ and $y\longmapsto \nu _{a}:=\nu
\circ y^{-1}$. Along the discontinuity surface $\Sigma $, the surface
displacement gradient is indicated by%
\begin{equation*}
\mathbb{W}:=\left\langle W\right\rangle \left( I-m\otimes m\right) .
\end{equation*}

The linearized kinematical setting justifies the mixed use in the same
context of symbols adopted elsewhere for distinct actual and referential
measures of interactions, as it is made in the ensuing section.

\section{Power and balance of interactions}

\subsection{Classification of the actions}

Relative changes of places between neighboring material elements generate
standard actions represented in the Lagrangian description (that is as a
field over $\mathcal{B}$) by the first Piola-Kirchhoff stress $P\in
Hom\left( T_{x}^{\ast }\mathcal{B},T_{y\left( x\right) }^{\ast }\mathcal{B}%
_{a}\right) $ and the vector of body forces $b\in \mathbb{R}^{3\ast }$.
Substructural events may occur within the material elements even when the
material elements themselves are frozen in space. Inhomogeneous
substructural changes in space generate new contact actions measured by the
so-called \emph{microstress} tensor $\mathcal{S}\in Hom\left( T_{x}^{\ast }%
\mathcal{B},T_{\nu \left( x\right) }^{\ast }\mathcal{M}\right) $. External
bulk fields can act directly over the substructures (magnetic and/or
electric fields, or some other radiative fields). They are represented by
the co-vector $\beta \in T_{\nu \left( x\right) }^{\ast }\mathcal{M}$, at
each $x$. All these interactions contribute to the expression of the power
of all external actions on a generic part of the body, namely on any subset $%
\mathfrak{b}$ of $\mathcal{B}$ with non vanishing volume and the same
regularity properties of $\mathcal{B}$ itself. It is said that a generic
part $\mathfrak{b}$ crosses $\Sigma $ - in this case it is indicated by $%
\mathfrak{b}_{\Sigma }$\ - when $\partial \mathfrak{b}_{\Sigma }\cap \Sigma $
is a simple closed curve where the normal \textsf{n} is defined as a vector
from $T_{x}^{\ast }\Sigma $ at all $x\in \partial \mathfrak{b}_{\Sigma }\cap
\Sigma $ where the normal to $\partial \mathfrak{b}_{\Sigma }$ exists; in
particular \textsf{n} at a given $x$ belongs to the tangent plane to $\Sigma 
$\ at the same point.

The surface $\Sigma $ can be considered as a model of a material layer with
vanishing thickness. In this case it is called a structured surface and it
is assumed that it can carry standard and substructural surface actions, the
former represented by a surface stress $\mathbb{T}$, the latter by a surface
microstress $\mathbb{S}$.

For a generic part $\mathfrak{b}_{\Sigma }$ crossing $\Sigma $, the explicit
expression of the power is the given by (see \cite{M02})%
\begin{eqnarray*}
\mathcal{P}_{\mathfrak{b}}^{ext}\left( \dot{y},\dot{\nu}\right) &:&=\int_{%
\mathfrak{b}_{\Sigma }}\left( b\cdot \dot{y}+\beta \cdot \dot{\nu}\right) 
\text{ }dx+\int_{\partial \mathfrak{b}_{\Sigma }}\left( Pn\cdot \dot{y}+%
\mathcal{S}n\cdot \dot{\nu}\right) \text{ }d\mathcal{H}^{2}+ \\
&&+\int_{\partial \mathfrak{b}_{\Sigma }\cap \Sigma }\left( \mathbb{T}%
\mathsf{n}\cdot \left\langle \dot{y}\right\rangle +\mathbb{S}\mathsf{n}\cdot
\left\langle \dot{\nu}\right\rangle \right) \text{ }d\mathcal{H}^{1}.
\end{eqnarray*}%
A Lagrangian representation is used in the earlier formula. The link with
the Eulerian (actual) description - a link given by the standard Piola
transform - is recalled later. Then everything is reduced to the linearized
setting.

\subsection{Observers}

An observer is intrinsically a representation of all geometrical
environments which are necessary to describe the morphology of a given body
and its motion.

The setting discussed here then incudes the assignment of atlantes over the
reference place $\mathcal{B}$, the ambient space $\mathbb{\hat{R}}^{3}$, the
interval of time and the manifold of substructural shapes $\mathcal{M}$.
Changes in such atlantes are changes in observers. Amid them the interest is
focused here on synchronous changes in observers - the ones leaving
invariant the representation of the time scale - which evaluate the same
reference place. In this sense only changes in $\mathbb{\hat{R}}^{3}$ and $%
\mathcal{M}$ are accounted for.

The ambient space $\mathbb{\hat{R}}^{3}$\ is altered by the action of the
group of diffeomorphisms onto itself, namely the group $Diff(\mathbb{\hat{R}}%
^{3},\mathbb{\hat{R}}^{3})$. Its action has infinitesimal generator
coinciding with the vector field which assigns to each point the vector $%
\frac{d\mathbf{f}_{s}}{ds}\left\vert _{s=0}\right. $, where $\mathbf{f}_{s}$
is a point selected over a smooth curve $s\longmapsto \mathbf{f}_{s}$, $s\in 
\mathbb{R}^{+}$, in $Diff(\mathbb{\hat{R}}^{3},\mathbb{\hat{R}}^{3})$ such
that $\mathbf{f}_{0}=identity$. The parameter $s$ can be identified with the
time.

Since the material substructures are in fact placed in space, changes of
frames in $\mathbb{\hat{R}}^{3}$ alter in principle the geometry of the
substructures and their consequent representation over $\mathcal{M}$. There
is exception when $\nu $ represents only a generic property of the material
substructure not associated with its geometry in space. Besides this
circumstance, one may presume the existence of an homomorphism $h:Diff(%
\mathbb{\hat{R}}^{3},\mathbb{\hat{R}}^{3})\rightarrow G,$ with $G$ the Lie
group of diffeomorphisms of $\mathcal{M}$ onto itself and $h$ mapping the
identity in $Diff(\mathbb{\hat{R}}^{3},\mathbb{\hat{R}}^{3})$ to the
identity in $G$. The curve $s\longmapsto \mathbf{f}_{s}$ then generates a
curve $s\longmapsto g_{s}:=h\left( \mathbf{f}_{s}\right) $ over $G$, and the
corresponding infinitesimal generator of the\ action of $G$ over $\mathcal{M}
$ is then defined by $\xi _{\mathcal{M}}\left( \nu \right) :=\frac{d\nu _{s}%
}{ds}\left\vert _{s=0}\right. =\frac{dh\left( \mathbf{f}_{s}\right) }{ds}%
\left\vert _{s=0}\right. $ (see related discussions in \cite{M08} and \cite%
{M08bis}).

Changes in observers generated by the group $SO\left( 3\right) $\ of the
proper rotations, a subgroup of $Diff(\mathbb{\hat{R}}^{3},\mathbb{\hat{R}}%
^{3})$ are specifically under scrutiny. For $\wedge q$ an element of the Lie
algebra $\mathfrak{so}\left( 3\right) $, $q\in \mathbb{\hat{R}}^{3}$, one
writes the corresponding $\xi _{\mathcal{M}}\left( \nu \right) $, obtained
through $h$, as the product $\mathcal{A}\left( \nu \right) q$ with $\mathcal{%
A}\left( \nu \right) \in Hom(\mathbb{\hat{R}}^{3},T_{\nu }\mathcal{M})$.

By indicating by $\dot{y}^{\ast }$ and $\dot{\nu}^{\ast }$ the pull-back in
the frame of the first observer of the rates evaluated by the second
observer), one gets%
\begin{equation*}
\dot{y}^{\ast }=\dot{y}+q\wedge \left( y-y_{0}\right)
\end{equation*}%
where $y_{0}$ is an arbitrarily fixed centre of rotation in the ambient
space, and%
\begin{equation*}
\dot{\nu}^{\ast }=\dot{\nu}+\mathcal{A}\left( \nu \right) q.
\end{equation*}%
Here $s$ is identified with the time.

\subsection{Invariance and its consequences}

\begin{axiom}
At (dynamic) equilibrium $\mathcal{P}_{\mathfrak{b}_{\Sigma }}^{ext}\left( 
\dot{y},\dot{\nu}\right) $ is invariant under rotational changes in
observers.
\end{axiom}

\begin{theorem}
(i) If for every $\mathfrak{b}_{\Sigma }$\ the vector fields assigning the
values $\sigma n$ and $\mathcal{A}^{\ast }\mathcal{S}n$ are defined over $%
\partial \mathfrak{b}_{\Sigma }$ and are integrable there, the integral
balances of actions on $\mathfrak{b}_{\Sigma }$%
\begin{equation*}
\int_{\mathfrak{b}_{\Sigma }}b\text{ }dx+\int_{\partial \mathfrak{b}_{\Sigma
}}Pn\text{ }d\mathcal{H}^{2}+\int_{\partial \mathfrak{b}_{\Sigma }\cap
\Sigma }\mathbb{T}n\text{ }d\mathcal{H}^{1}=0,
\end{equation*}%
\begin{equation*}
\int_{\mathfrak{b}_{\Sigma }}\left( \left( x-x_{0}\right) \wedge b+\mathcal{A%
}^{\ast }\beta \right) \text{ }dx+\int_{\partial \mathfrak{b}_{\Sigma
}}\left( \left( x-x_{0}\right) \wedge Pn+\mathcal{A}^{\ast }\mathcal{S}%
n\right) \text{ }d\mathcal{H}^{2}+
\end{equation*}%
\begin{equation*}
+\int_{\partial \mathfrak{b}_{\Sigma }\cap \Sigma }\left( \left(
y-y_{0}\right) \wedge \mathbb{T}\mathsf{n}+\mathcal{A}^{\ast }\mathbb{S}%
\mathsf{n}\right) \text{ }d\mathcal{H}^{1}=0,
\end{equation*}%
(ii) Moreover, if the tensor fields $x\longmapsto $ $P$, $\mathcal{S}$ are
of class $C^{1}\left( \mathcal{B}\backslash \Sigma \right) \ $and are also
continuous over the boundary of the body, then%
\begin{equation*}
DivP+b=0,
\end{equation*}%
and there exist a co-vector field $x\longmapsto z\in T_{\nu \left( x\right) }%
\mathcal{M}$ such that%
\begin{equation*}
skw\left( PF^{\ast }\right) =\frac{1}{2}\mathsf{e}\left( \mathcal{A}^{\ast
}z+\left( \nabla \mathcal{A}^{\ast }\right) \mathcal{S}\right)
\end{equation*}%
and%
\begin{equation*}
Div\mathcal{S}-z+\beta =0,
\end{equation*}%
with $z=z_{1}+z_{2}$, $z_{2}\in Ker\mathcal{A}^{\ast }$in the bulk.
Additionally, if the tensor fields $x\longmapsto $, $\mathbb{T},\mathbb{S}$
are of class $C^{1}\left( \Sigma \right) $\ along the surface $\Sigma $ and
are also continuous along its boundary, one gets%
\begin{equation*}
Div_{\Sigma }\mathbb{T}+[P]m=0,
\end{equation*}%
and there exists a co-vector field $x\longmapsto \mathfrak{z}\in T_{\nu
\left( x\right) }\mathcal{M}$, with $x\in \Sigma $, such that%
\begin{equation*}
skw\left( \mathbb{TF}^{\ast }\right) =\frac{1}{2}\mathsf{e}\left( \mathcal{A}%
^{\ast }\mathfrak{z}+\left( \nabla _{\Sigma }\mathcal{A}^{\ast }\right) 
\mathbb{S}\right)
\end{equation*}%
and%
\begin{equation*}
Div_{\Sigma }\mathbb{S}-\mathfrak{z}+[\mathcal{S}]m=0.
\end{equation*}%
(iii) If the rate fields $\left( x,t\right) \longmapsto \dot{y}\left(
x,t\right) \in \mathbb{\hat{R}}^{3}$ and $\left( x,t\right) \longmapsto \dot{%
\nu}\left( x,t\right) \in T_{\nu \left( x\right) }\mathcal{M}$ are
differentiable in space, the local balances imply%
\begin{equation*}
\mathcal{P}_{\mathfrak{b}}^{ext}\left( \dot{y},\dot{\nu}\right) =\mathcal{P}%
_{\mathfrak{b}}^{int}\left( \dot{y},\dot{\nu}\right)
\end{equation*}%
where%
\begin{equation*}
\mathcal{P}_{\mathfrak{b}}^{int}\left( \dot{y},\dot{\nu}\right) :=\int_{%
\mathfrak{b}}(P\cdot \dot{F}+z\cdot \dot{\nu}+\mathcal{S}\cdot \dot{N})\text{
}dx+\int_{\partial \mathfrak{b}_{\Sigma }\cap \Sigma }\left( \mathbb{T\cdot F%
}+\mathfrak{z}\cdot \dot{\nu}+\mathbb{S\cdot N}\right) \text{ }d\mathcal{H}%
^{2}.
\end{equation*}
\end{theorem}

$\mathcal{P}_{\mathfrak{b}}^{int}\left( \dot{y},\dot{\nu}\right) $\ is
called an inner (or internal) power. \textsf{e} indicates Ricci's
alternating index.

If one enforces Axiom 1 with a requirement of invariance with respect to the
action of the semi-direct product $\mathbb{\hat{R}}^{3}\ltimes SO\left(
3\right) $, rather that calling upon only the action of $SO\left( 3\right) $%
, a proof of Theorem 1 can be found in \cite{M02}. The weaker requirement
here imposes a change in the proof. By using the Axiom 1, in fact, one first
obtains only the integral balance of moments, then one first exploits the
arbitrariness of the centre of rotation and substitutes $y_{0}$ with $%
y_{0}+w $, with $w$ an arbitrary vector depending only on time. If one
subtracts the integral balance of moments from the resulting equation (the
one obtained by the substitution $y_{0}\longmapsto y_{0}+w$) and the
arbitrariness of $w$ also imply the integral balance of forces. Pointwise
balances follow by the standard use of Gauss theorem (see also remarks in 
\cite{M08, M08bis}).

In the Eulerian representation the balance equations become%
\begin{equation*}
div\sigma +b_{a}=0,
\end{equation*}%
\begin{equation*}
skw\left( \sigma \right) =\frac{1}{2}\mathsf{e}\left( \mathcal{A}^{\ast
}z_{a}+\left( grad\mathcal{A}^{\ast }\right) \mathcal{S}_{a}\right) ,
\end{equation*}%
\begin{equation*}
div\mathcal{S}_{a}-z_{a}+\beta _{a}=0,
\end{equation*}%
in the bulk and%
\begin{equation*}
div_{\Sigma }\mathbb{T}_{a}+[\sigma ]m_{a}=0,
\end{equation*}%
\begin{equation*}
skw\left( \mathbb{T}_{a}\right) =\frac{1}{2}\mathsf{e}\left( \mathcal{A}%
^{\ast }\mathfrak{z}_{a}+\left( grad_{\Sigma }\mathcal{A}^{\ast }\right) 
\mathbb{S}_{a}\right) ,
\end{equation*}%
\begin{equation*}
div_{\Sigma }\mathbb{S}_{a}-\mathfrak{z}_{a}+[\mathcal{S}_{a}]m_{a}=0,
\end{equation*}%
over $\Sigma $. Moreover, by indicating by $v$ and $\upsilon $ the values
the values of the actual (Eulerian) representation of the (differentiable)
velocity fields $\left( y,t\right) \longmapsto v\left( y,t\right) $ and $%
\left( y,t\right) \longmapsto \upsilon \left( y,t\right) $, one also get the
equation%
\begin{equation*}
\int_{y\left( \mathfrak{b}\right) }\left( b_{a}\cdot v+\beta \cdot \upsilon
\right) \text{ }dy+\int_{\partial y\left( \mathfrak{b}\right) }\left( \sigma
n_{a}\cdot v+\mathcal{S}_{a}n_{a}\cdot \upsilon \right) \text{ }d\mathcal{H}%
^{2}+
\end{equation*}%
\begin{equation*}
+\int_{\partial y\left( \mathfrak{b}\right) \cap y\left( \Sigma \right)
}\left( \mathbb{T}_{a}\mathsf{n}_{a}\cdot \left\langle v\right\rangle +%
\mathbb{S}_{a}\mathsf{n}_{a}\cdot \left\langle \upsilon \right\rangle
\right) \text{ }d\mathcal{H}^{1}=\int_{y\left( \mathfrak{b}\right) }\left(
\sigma \cdot gradv+z_{a}\cdot \upsilon +\mathcal{S}_{a}\cdot grad\upsilon
\right) \text{ }dy+
\end{equation*}%
\begin{equation*}
+\int_{y\left( \mathfrak{b}\right) \cap y\left( \Sigma \right) }\left( 
\mathbb{T}_{a}\cdot grad_{\Sigma }v+\mathfrak{z}_{a}\cdot \upsilon +\mathbb{S%
}_{a}\cdot grad_{\Sigma }\upsilon \right) \text{ }d\mathcal{H}^{2}.
\end{equation*}%
In previous formulas $n_{a}$, $m_{a}$ and $\mathsf{n}_{a}$\ are the
counterparts of $n$, $m$ and $\mathsf{n}$\ in $\mathcal{B}_{a}$.\ The
differential operators $div$, $div_{\Sigma }$, $grad$ and $grad_{\Sigma }$
involve derivatives with respect to the coordinates $y^{i}$ in the current
macroscopic placement. The index $a$ means `actual'. Moreover, the actual
measures of interactions are obtained by means of the standard Piola
transform as%
\begin{equation*}
b_{a}:=\left( \det F\right) ^{-1}b,\text{ \ \ }\sigma :=\left( \det F\right)
^{-1}PF^{\ast },
\end{equation*}%
\begin{equation*}
z_{a}:=\left( \det F\right) ^{-1}z,\text{ \ \ }\beta _{a}:=\left( \det
F\right) ^{-1}\beta ,\text{ \ \ }\mathcal{S}_{a}:=\left( \det F\right) ^{-1}%
\mathcal{S}F^{\ast },
\end{equation*}%
\begin{equation*}
\mathbb{T}_{a}:=\left( \det \mathbb{F}\right) ^{-1}\mathbb{TF}^{\ast },\text{
\ \ }\mathfrak{z}_{a}:=\left( \det \mathbb{F}\right) ^{-1}\mathfrak{z},\text{
\ \ }\mathbb{S}_{a}:=\left( \det \mathbb{F}\right) ^{-1}\mathbb{SF}^{\ast }.
\end{equation*}%
The proof of the Piola transform can be found on any textbook in nonlinear
continuum mechanics. Less popular is its counterpart on surfaces embedded in
a body: for the relevant proof see \cite{GuMu}.

The Piola transform implies that, in infinitesimal deformation setting,
referential and actual measures of interaction in the bulk and over the
surface $\Sigma $ coincide as $\left\vert W\right\vert $ and $\left\vert 
\mathbb{W}\right\vert $ tend to zero.

The ensuing sections are just restricted to the infinitesimal deformation
setting. Thus, by taking into account the substantial coincidence of
referential and actual measures of interactions in the linearized setting
(as remarked in earlier comments), although the Cauchy stress $\sigma $ is
used, the index "$a$" in the other actual measures of interaction is omitted
for the sake of conciseness.

\section{Linear constitutive structures}

Constitutive structures of the type%
\begin{equation*}
\sigma =\sigma \left( W,\nu ,N\right) ,\text{ \ \ }z=z\left( W,\nu ,N\right)
,\text{ \ \ }\mathcal{S}=\mathcal{S}\left( W,\nu ,N\right) ,
\end{equation*}%
in the bulk and%
\begin{equation*}
\mathbb{T}=\mathbb{T}\left( \mathbb{W},\nu ,\mathbb{N}\right) ,\text{ \ \ }%
\mathfrak{z}=\mathfrak{z}\left( \mathbb{W},\nu ,\mathbb{N}\right) ,\text{ \
\ }\mathbb{S}=\mathbb{S}\left( \mathbb{W},\nu ,\mathbb{N}\right) ,
\end{equation*}%
on the surface $\Sigma $\ can be selected for elastic complex bodies. The
entries of the previous constitutive structures are instantaneous value.
Linearization of them about a pair $\left( \bar{u}\mathbf{,}\bar{\nu}\right) 
$ requires the embedding of $\mathcal{M}$ in some linear space isomorphic to 
$\mathbb{R}^{k}$ for some $k$. This embedding allows one to consider the
space of pairs of maps $\left( y,\nu \right) $ as an infinite-dimensional
manifold modelled over a Sobolev space: the use of the Frechet derivative in
the linearization procedure then follows.

It is then assumed that $\mathcal{M}$ is endowed with a $C^{1}$ Riemannian
metric and the relevant Levi-Civita parallel transport. Such structural
assumption is constitutive in its essential nature. By Nash theorem an
isometric embedding of $\mathcal{M}$\ in a linear space is then always
available but it is neither unique nor rigid. The selection of an embedding
plays the role of a constitutive ingredient of every special model.

Under these conditions, constitutive equations expressed by linear operators 
$\mathbf{L}^{\left( \cdot \right) }$, namely%
\begin{equation*}
\sigma =\mathbf{L}^{\left( \sigma \right) }\left( W,\nu ,N\right) ,\text{ \
\ }z=\mathbf{L}^{\left( z\right) }\left( W,\nu ,N\right) ,\text{ \ \ }%
\mathcal{S}=\mathbf{L}^{\left( \mathcal{S}\right) }\left( W,\nu ,N\right) ,
\end{equation*}%
in the bulk and 
\begin{equation*}
\mathbb{T}=\mathbf{L}^{\left( \mathbb{T}\right) }\left( \mathbb{W},\nu ,%
\mathbb{N}\right) ,\text{ \ \ }\mathfrak{z}=\mathbf{L}^{\left( \mathfrak{z}%
\right) }\left( \mathbb{W},\nu ,\mathbb{N}\right) ,\text{ \ \ }\mathbb{S}=%
\mathbf{L}^{\left( \mathbb{S}\right) }\left( \mathbb{W},\nu ,\mathbb{N}%
\right) ,
\end{equation*}%
on the discontinuity surface, make sense.

The procedure discussed here holds also when the measures of interaction
depend not only on instantaneous values of the state variables but also on
their \emph{entire history}. In this case memory effects can be accounted
for. Viscosity come into play.

\section{Characterization of the bulk free energy in terms of work:
variations on a Del Piero's theme}

\subsection{Histories}

At a point $x$, a \emph{history} is a $BV$ right continuous map%
\begin{equation*}
H:\mathbb{R}^{+}\rightarrow M_{3\times 3}\times \mathbb{R}^{k}\times
M_{k\times 3}
\end{equation*}%
such that, for $s\in \mathbb{R}^{+}$,%
\begin{equation*}
H(s)=\left( W(s),\nu (s),N(s)\right) .
\end{equation*}%
Its restriction $K_{p}^{r}$\ over an interval $[r,p)$ is called \emph{process%
} and is defined by%
\begin{equation*}
K_{p}^{r}(s):=\left( W(r+s),\nu (r+s),N(r+s)\right) ,\quad 0\leq s<p-r.
\end{equation*}%
As shorthand notation, $K_{p}$ indicates a process when it is of the type $%
K_{p}^{0}$. The symbols $\Gamma $ and $\Pi $ denote the spaces of histories
and that of processes, respectively. Of course, here $\Pi \subseteq \Gamma $%
. Processes prolong histories. Given $H$, the history%
\begin{equation}
\left( K_{p}\ast H\right) (s):=\left\{ 
\begin{array}{ll}
K_{p}(s) & 0\leq s<p \\ 
H(s-p) & s\geq p%
\end{array}%
\right.  \label{Cont}
\end{equation}%
is called \emph{prolongation of }$H$ \emph{by means of the process }$K_{p}$.
It is assumed also that%
\begin{equation*}
K_{p}\left( p\right) ^{-}=H\left( 0\right) ,
\end{equation*}%
where $K_{p}\left( p\right) ^{-}:=\lim_{s\nearrow p}K_{p}\left( s\right) $,
to assure differentiability in time, a property necessary for later use.

Along $\Sigma $, a surface history%
\begin{equation*}
s\longmapsto \mathbb{H}(s):=(\mathbb{W}(s),\nu (s),\mathbb{N}(s)),
\end{equation*}%
can be defined when the map $x\longmapsto \nu (x)$\ is continuous across the
surface. Of course, since $\mathcal{M}$ is embedded in a linear space, the
average $\left\langle \nu \right\rangle $ makes now sense so that one may
consider also $\mathbb{H}(s)$\ to be coincident with $(\mathbb{W}%
(s),\left\langle \nu \right\rangle (s),\mathbb{N}(s))$. The results
collected below hold also in this case.

\subsection{History dependent measures of interaction and equivalence of
histories}

In the earlier notes, the state variables have been indicates just formally.
When memory effects are accounted for, the notion of state requires careful
definition because different equivalence relations between pairs of
histories exist so that the notion of state appears to be rather natural in
terms of equivalence classes. Such a question has been tackled variously
(see \cite{DPD, DP, FGM}) on the basis of the abstract approach to
thermodynamics proposed in \cite{CO1} (see also \cite{COS, Si1, Si2}).

In what follows, the point of view developed by Del Piero in \cite{DP} for
linear simple bodies with memory is adapted to cover linear complex bodies
displaying memory effects at macroscopic and microscopic scales. The aim is
of (1) characterizing states and (2) deducing the main property of the free
energy of such complex bodies. A concrete example of such bodies is the one
of relaxor ferroelectrics. Butterfly loops in the diagrams of strain versus
applied electric fields indicate the presence of memory effects \cite{VuR}.

Linear constitutive structures in the bulk are here assumed to be of the
form 
\begin{eqnarray*}
\sigma \left( H\right) &=&G_{\sigma W}\left( 0\right) W\left( 0\right)
+G_{\sigma \nu }\left( 0\right) \nu \left( 0\right) +G_{\sigma N}\left(
0\right) N\left( 0\right) + \\
&&+\int_{0}^{+\infty }\left( \dot{G}_{\sigma W}\left( s\right) W\left(
s\right) +\dot{G}_{\sigma \nu }\left( s\right) \nu \left( s\right) +\dot{G}%
_{\sigma N}\left( s\right) N\left( s\right) \right) \text{ }ds,
\end{eqnarray*}%
\begin{eqnarray*}
\mathcal{S}\left( H\right) &=&G_{\mathcal{S}W}\left( 0\right) W\left(
0\right) +G_{\mathcal{S}\nu }\left( 0\right) \nu \left( 0\right) +G_{%
\mathcal{S}N}\left( 0\right) N\left( 0\right) + \\
&&+\int_{0}^{+\infty }\left( \dot{G}_{\mathcal{S}W}\left( s\right) W\left(
s\right) +\dot{G}_{\mathcal{S}\nu }\left( s\right) \nu \left( s\right) +\dot{%
G}_{\mathcal{S}N}\left( s\right) N\left( s\right) \right) \text{ }ds,
\end{eqnarray*}%
\begin{eqnarray*}
z\left( H\right) &=&G_{zW}\left( 0\right) W\left( 0\right) +G_{z\nu }\left(
0\right) \nu \left( 0\right) +G_{zN}\left( 0\right) N\left( 0\right) + \\
&&+\int_{0}^{+\infty }\left( \dot{G}_{zW}\left( s\right) W\left( s\right) +%
\dot{G}_{z\nu }\left( s\right) \nu \left( s\right) +\dot{G}_{zN}\left(
s\right) N\left( s\right) \right) \text{ }ds,
\end{eqnarray*}%
where the $G_{AB}$'s are the so-called \emph{relaxation functions}, tensor
functions (taking values in different tensor spaces) that are assumed to be
Lebesgue integrable in time: they are absolutely continuous and the limit $%
G_{AB}(\infty ):=\lim_{s\rightarrow +\infty }G_{AB}(s)$ exists. Of course,
the indexes $A$ and $B$ run in $\{\sigma ,z,S\}$ and $\{W,\nu ,N\}$
respectively.

Notice that the integrals above are well defined because $\mathcal{M}$ is
considered embedded in a linear space isomorphic to $\mathbb{R}^{k}$ for
some $k$.

\begin{definition}
Two generic histories $H$ and $H^{\prime }$, such that $H(0)=H^{\prime }(0)$%
,\ are said to be equivalent (in symbols $H\sim H^{\prime }$) when for every
process $K_{p}$, with $p\geq 0$,%
\begin{equation*}
\sigma \left( K_{p}\ast H\right) =\sigma \left( K_{p}\ast H^{\prime }\right)
,\text{ }z\left( K_{p}\ast H\right) =z\left( K_{p}\ast H^{\prime }\right) ,%
\text{ }\mathcal{S}\left( K_{p}\ast H\right) =\mathcal{S}\left( K_{p}\ast
H^{\prime }\right) .
\end{equation*}
\end{definition}

As a consequence, for $p\geq 0$, the condition of equivalence $H\sim
H^{\prime }$ implies%
\begin{equation*}
\int_{0}^{+\infty }\left( \dot{G}_{\sigma W}\left( s+p\right) W\left(
s\right) +\dot{G}_{\sigma \nu }\left( s+p\right) \nu \left( s\right) +\dot{G}%
_{\sigma N}\left( s+p\right) N\left( s\right) \right) \text{ }ds=
\end{equation*}%
\begin{equation*}
=\int_{0}^{+\infty }\left( \dot{G}_{\sigma W}\left( s+p\right) W^{\prime
}\left( s\right) +\dot{G}_{\sigma \nu }\left( s+p\right) \nu ^{\prime
}\left( s\right) +\dot{G}_{\sigma N}\left( s+p\right) N^{\prime }\left(
s\right) \right) \text{ }ds
\end{equation*}%
\begin{equation*}
\int_{0}^{+\infty }\left( \dot{G}_{\mathcal{S}W}\left( s+p\right) W\left(
s\right) +\dot{G}_{\mathcal{S}\nu }\left( s+p\right) \nu \left( s\right) +%
\dot{G}_{\mathcal{S}N}\left( s+p\right) N\left( s\right) \right) \text{ }ds=
\end{equation*}%
\begin{equation*}
=\int_{0}^{+\infty }\left( \dot{G}_{\mathcal{S}W}\left( s+p\right) W^{\prime
}\left( s\right) +\dot{G}_{\mathcal{S}\nu }\left( s+p\right) \nu ^{\prime
}\left( s\right) +\dot{G}_{\mathcal{S}N}\left( s+p\right) N^{\prime }\left(
s\right) \right) \text{ }ds
\end{equation*}%
\begin{equation*}
\int_{0}^{+\infty }\left( \dot{G}_{zW}\left( s+p\right) W\left( s\right) +%
\dot{G}_{z\nu }\left( s+p\right) \nu \left( s\right) +\dot{G}_{zN}\left(
s+p\right) N\left( s\right) \right) \text{ }ds=
\end{equation*}%
\begin{equation*}
=\int_{0}^{+\infty }\left( \dot{G}_{zW}\left( s+p\right) W^{\prime }\left(
s\right) +\dot{G}_{z\nu }\left( s+p\right) \nu ^{\prime }\left( s\right) +%
\dot{G}_{zN}\left( s+p\right) N^{\prime }\left( s\right) \right) \text{ }ds.
\end{equation*}

Let the distance $d\left( H,H^{\prime }\right) $ be defined by%
\begin{equation*}
d\left( H,H^{\prime }\right) :=\sup_{t>0}\left\{ \left\vert
\int_{0}^{+\infty }\left( \dot{G}_{\sigma W}\left( s+t\right) W\left(
s\right) +\dot{G}_{\sigma \nu }\left( s+t\right) \nu \left( s\right) +\dot{G}%
_{\sigma N}\left( s+t\right) N\left( s\right) \right) \text{ }ds\right.
\right. -
\end{equation*}%
\begin{equation*}
-\left. \int_{0}^{+\infty }\left( \dot{G}_{\sigma W}\left( s+t\right)
W^{\prime }\left( s\right) +\dot{G}_{\sigma \nu }\left( s+t\right) \nu
^{\prime }\left( s\right) +\dot{G}_{\sigma N}\left( s+t\right) N^{\prime
}\left( s\right) \right) \text{ }ds\right\vert +
\end{equation*}%
\begin{equation*}
+\left\vert \int_{0}^{+\infty }\left( \dot{G}_{zW}\left( s+t\right) W\left(
s\right) +\dot{G}_{z\nu }\left( s+t\right) \nu \left( s\right) +\dot{G}%
_{zN}\left( s+t\right) N\left( s\right) \right) \text{ }ds\right. -
\end{equation*}%
\begin{equation*}
-\left. \int_{0}^{+\infty }\left( \dot{G}_{zW}\left( s+t\right) W^{\prime
}\left( s\right) +\dot{G}_{z\nu }\left( s+t\right) \nu ^{\prime }\left(
s\right) +\dot{G}_{zN}\left( s+t\right) N^{\prime }\left( s\right) \right) 
\text{ }ds\right\vert +
\end{equation*}%
\begin{equation*}
+\left\vert \int_{0}^{+\infty }\left( \dot{G}_{\mathcal{S}W}\left(
s+t\right) W\left( s\right) +\dot{G}_{\mathcal{S}\nu }\left( s+t\right) \nu
\left( s\right) +\dot{G}_{\mathcal{S}N}\left( s+t\right) N\left( s\right)
\right) \text{ }ds\right. -
\end{equation*}%
\begin{equation*}
-\left. \left. \int_{0}^{+\infty }\left( \dot{G}_{\mathcal{S}W}\left(
s+t\right) W^{\prime }\left( s\right) +\dot{G}_{\mathcal{S}\nu }\left(
s+t\right) \nu ^{\prime }\left( s\right) +\dot{G}_{\mathcal{S}N}\left(
s+t\right) N^{\prime }\left( s\right) \right) \text{ }ds\right\vert \right\}
.
\end{equation*}%
It induces a pseudometric in the space $\Gamma $ of histories and a metric
in the quotient space $\Gamma /\sim $. By taking the limit $t\rightarrow 0$
in the expression above, one gets%
\begin{equation*}
\left\vert \sigma \left( H\right) -\sigma \left( H^{\prime }\right)
\right\vert +\left\vert z\left( H\right) -z\left( H^{\prime }\right)
\right\vert +\left\vert \mathcal{S}\left( H\right) -\mathcal{S}\left(
H^{\prime }\right) \right\vert \leq d\left( H,H^{\prime }\right) +
\end{equation*}%
\begin{equation*}
\left\vert \left( G_{\sigma W}\left( 0\right) W\left( 0\right) +G_{\sigma
\nu }\left( 0\right) \nu \left( 0\right) +G_{\sigma N}\left( 0\right)
N\left( 0\right) \right) \right. -
\end{equation*}%
\begin{equation*}
\left. \left( G_{\sigma W}\left( 0\right) W^{\prime }\left( 0\right)
+G_{\sigma \nu }\left( 0\right) \nu ^{\prime }\left( 0\right) +G_{\sigma
N}\left( 0\right) N^{\prime }\left( 0\right) \right) \right\vert +
\end{equation*}%
\begin{equation*}
\left\vert \left( G_{zW}\left( 0\right) W\left( 0\right) +G_{z\nu }\left(
0\right) \nu \left( 0\right) +G_{zN}\left( 0\right) N\left( 0\right) \right)
\right. -
\end{equation*}%
\begin{equation*}
\left. \left( G_{zW}\left( 0\right) W^{\prime }\left( 0\right) +G_{z\nu
}\left( 0\right) \nu ^{\prime }\left( 0\right) +G_{zN}\left( 0\right)
N^{\prime }\left( 0\right) \right) \right\vert +
\end{equation*}%
\begin{equation*}
\left\vert \left( G_{\mathcal{S}W}\left( 0\right) W\left( 0\right) +G_{%
\mathcal{S}\nu }\left( 0\right) \nu \left( 0\right) +G_{\mathcal{S}N}\left(
0\right) N\left( 0\right) \right) \right. -
\end{equation*}%
\begin{equation*}
\left. \left( G_{\mathcal{S}W}\left( 0\right) W^{\prime }\left( 0\right) +G_{%
\mathcal{S}\nu }\left( 0\right) \nu ^{\prime }\left( 0\right) +G_{\mathcal{S}%
N}\left( 0\right) N^{\prime }\left( 0\right) \right) \right\vert .
\end{equation*}%
Since two equivalent histories are characterized by identical initial
values, namely $W\left( 0\right) =$ $W^{\prime }\left( 0\right) $, $\nu
\left( 0\right) =\nu ^{\prime }\left( 0\right) $, $N\left( 0\right)
=N^{\prime }\left( 0\right) $, from previous inequality it follows that 
\emph{two equivalent histories determine the same macroscopic stress,
microstress and substructural self-action}. The proof of the analogous
property for simple bodies in \cite{DP} is based on the use of a seminorm.

\begin{proposition}
The pseudometric $d\left( \cdot ,\cdot \right) $ has the following
properties:
\end{proposition}

\begin{description}
\item[Contraction] \emph{for every} $p>0$%
\begin{equation*}
d\left( K_{p}\ast H,K_{p}\ast H^{\prime }\right) \leq d\left( H,H^{\prime
}\right) .
\end{equation*}

\item[Fading memory] \emph{for every} $\varepsilon >0$ \emph{there exists} $%
r $ \emph{such that, for every} $p>r$\emph{, one gets}%
\begin{equation*}
d\left( K_{p}\ast H,K_{p}\ast H^{\prime }\right) <\varepsilon .
\end{equation*}

\item[Approachability] \emph{if }$H_{p}\left( p\right) ^{-}=H^{\prime
}\left( 0\right) $\emph{, then}%
\begin{equation*}
\lim_{d\rightarrow +\infty }d\left( H,H_{p}\ast H^{\prime }\right) =0,
\end{equation*}%
\emph{with} $H_{p}$ \emph{the process generated by} $H$ \emph{over} $[0,p)$.
\end{description}

\begin{proof}
Contraction arises directly from the definition. In fact, in writing
explicitly $d\left( K_{p}\ast H,K_{p}\ast H^{\prime }\right) $ by taking
into account (\ref{Cont}) and the definition of the semimetric $d\left(
\cdot ,\cdot \right) $, one manages integrals containing terms of the type $%
\dot{G}_{AB}\left( s+t\right) B\left( s-r\right) $, with $A$ and $B$ running
in $\{\sigma ,z,S\}$ and $\{W,\nu ,N\}$ respectively. Each of these terms is
also equal to $\dot{G}_{AB}\left( s+t+r\right) B\left( s\right) $\ (see also 
\cite{DP}). Consequently, by taking the superemum over $t$, $p>0$ one gets
the contraction property straight away.

Let now define%
\begin{equation*}
M:=\max \left\{ \sup_{s>0}|W^{\prime }(s)|,\sup_{s>0}|W(s)|,\sup_{s>0}|\nu
^{\prime }(s)|,\sup_{s>0}|\nu (s)|,\sup_{s>0}|N^{\prime
}(s)|,\sup_{s>0}|N(s)|\right\} .
\end{equation*}%
By taking into account that%
\begin{equation*}
\left\vert \int_{a}^{b}\dot{G}_{AB}(s)B(s)\text{ }ds\right\vert \leq
\sup_{s\in \lbrack a,b]}\left\vert B(s)\right\vert \int_{a}^{b}\left\vert
G_{AB}(s)\right\vert \text{ }ds,
\end{equation*}%
one obtains%
\begin{equation*}
d\left( K_{p}\ast H,K_{p}\ast H^{\prime }\right) \leq 2M\sup_{t>0}\left\{
\int_{p}^{+\infty }\left\vert \dot{G}_{\sigma W}\left( s+t\right)
\right\vert \text{ }ds\right. +
\end{equation*}%
\begin{equation*}
+\int_{p}^{+\infty }\left\vert \dot{G}_{\sigma \nu }\left( s+t\right)
\right\vert \text{ }ds+\int_{p}^{+\infty }\left\vert \dot{G}_{\sigma
N}\left( s+t\right) \right\vert \text{ }ds+
\end{equation*}%
\begin{equation*}
+\int_{p}^{+\infty }\left\vert \dot{G}_{zW}\left( s+t\right) \right\vert 
\text{ }ds+\int_{p}^{+\infty }\left\vert \dot{G}_{z\nu }\left( s+t\right)
\right\vert \text{ }ds+\int_{p}^{+\infty }\left\vert \dot{G}_{zN}\left(
s+t\right) \right\vert \text{ }ds+
\end{equation*}%
\begin{equation*}
\left. +\int_{p}^{+\infty }\left\vert \dot{G}_{\mathcal{S}W}\left(
s+t\right) \right\vert \text{ }ds+\int_{p}^{+\infty }\left\vert \dot{G}_{%
\mathcal{S}\nu }\left( s+t\right) \right\vert \text{ }ds+\int_{p}^{+\infty
}\left\vert \dot{G}_{\mathcal{S}N}\left( s+t\right) \right\vert \text{ }%
ds\right\} .
\end{equation*}%
However, since it has been assumed that the maps $s\longmapsto |\dot{G}%
_{AB}|\left( s\right) $ are integrable, there exists $r$ such that, for $r>m$%
, one may find $\varepsilon >0$ such that the right-hand side of the
previous relation is lesser or equal to $\varepsilon \left( 2M\right) ^{-1}$%
. Fading memory then follows. It also implies the property of
approachability. In fact, from%
\begin{equation*}
d(H_{p}\ast H^{\prime },H)=d(H_{p}\ast H^{\prime },H_{p}\ast H^{d})
\end{equation*}%
and fading memory, the approachability can be obtained by letting $p$ to $%
+\infty $.
\end{proof}

Previous theorem suggests the following definition:

\begin{definition}
$H$ is said to be approachable from another history $H^{\prime }$ if there
exists a family of processes $\left( p\longmapsto K_{p},p\in \mathbb{R}%
^{+}\right) $ prolonging $H^{\prime }$ and such that $K_{p}\ast H^{\prime }$
converges to $H^{\prime }$ with respect to the pseudometric $d\left( \cdot
,\cdot \right) $ as $p\rightarrow +\infty $.
\end{definition}

\subsection{States and actions}

The state space is identified here with the space of histories $\Gamma $
endowed with the norm%
\begin{equation*}
\left\Vert H\right\Vert _{\Gamma }=\left\vert F\right\vert
_{L^{2}}+\left\vert \nu \right\vert _{W^{1,2}}.
\end{equation*}%
In this way each state is defined to within an equivalent history.

\begin{definition}
A function $f:\Gamma \rightarrow \mathbb{R}$ is called a \textbf{state
function} if $H\sim H^{\prime }$ implies $f\left( H\right) =f\left(
H^{\prime }\right) $.
\end{definition}

\begin{definition}
A function $a:\Gamma \times \Pi \rightarrow \mathbb{R}$ is called\ an 
\textbf{action} if
\end{definition}

\begin{description}
\item[(1)] $a$ \emph{is additive with respect to prolongations, namely}%
\begin{equation*}
a\left( K^{\prime }\ast K,H\right) =a\left( K^{\prime },K\ast H\right)
+a\left( K,H\right) ,
\end{equation*}

\item[(2)] \emph{the map} $a\left( K,\cdot \right) :\Gamma \rightarrow 
\mathbb{R}$ \emph{is continuous.}
\end{description}

\begin{definition}
The action $a$ satisfies the \textbf{dissipation property} along $H$ if, for
every $\varepsilon >0$, there exists $\delta >0$ such that%
\begin{equation*}
d\left( K_{p}\ast H,H\right) <\delta \Longrightarrow a\left( K_{p},H\right)
>-\varepsilon .
\end{equation*}
\end{definition}

\begin{definition}
Given a generic action $a\left( \cdot ,\cdot \right) $, a function $f:\Gamma
\rightarrow \mathbb{R}$ is called a \textbf{lower potential} for $a\left(
\cdot ,\cdot \right) $ if for every $H$ and $H^{\prime }$ belonging to $%
\Gamma $ and for every $\varepsilon >0$ there exists $\delta >0$ such that%
\begin{equation*}
f\left( H\right) -f\left( H^{\prime }\right) <a\left( K,H\right)
+\varepsilon ,
\end{equation*}%
for every $K$ such that $d\left( K\ast H^{\prime },H\right) <\delta $.
\end{definition}

The definition of state function comes from \cite{DP} - the difference is
here only the extended meaning of the history -\ while the definitions of
action and lower potential have been introduced in \cite{CO1}.

\subsection{Work density in the bulk}

At every $x\in \mathcal{B}$ the \emph{work density} $w\left( H\right) $ is
defined by%
\begin{equation*}
w(H):=\int_{0}^{+\infty }\left[ \sigma \left( H^{s}\right) \cdot \dot{F}%
(s)+z\left( H^{s}\right) \cdot \dot{\nu}(s)+\mathcal{S}\left( H^{s}\right)
\cdot \dot{N}(s)\right] \text{ }ds.
\end{equation*}%
Note the different algebraic sign in the analogous definition of $w\left(
H\right) $ given in \cite{DP} with reference to simple bodies, i.e. in
absence of substructural interactions.

Given an history $H$ and a prolongation $K_{p}=\left( F_{p},\nu
_{p},N_{p}\right) $ of it, \emph{the work density over the prolongation} $%
K_{p}$, indicated by $w\left( K_{p},H\right) $ is then defined by%
\begin{equation*}
w\left( K_{p},H\right) :=w\left( K_{p}\ast H\right) -w(H)
\end{equation*}%
that is%
\begin{equation*}
w\left( K_{p},H\right) =\int_{0}^{p}\left[ \sigma \left( K_{p}^{s}\ast
H\right) \cdot \dot{W}_{p}(s)+z\left( K_{p}^{s}\ast H\right) \cdot \dot{\nu}%
_{p}(s)+\mathcal{S}\left( K_{p}^{s}\ast H\right) \cdot \dot{N}_{p}(s)\right] 
\text{ }ds.
\end{equation*}%
An analogous power density over prolongations is defined in \cite{DP} but
with reference to relative continuations (defined below). The use of
relative prolongation implies the appearance of further terms in the
explicit expression of $w\left( K_{d},H\right) $, terms due to the jump in $%
H\left( 0\right) $.

\begin{theorem}
The work density over a prolongation is an action and for any fixed $K_{p}$
the map $w\left( K_{p},\cdot \right) $ is a state function.
\end{theorem}

\begin{proof}
To prove the latter property, first define%
\begin{equation*}
\tilde{M}:=\sup \left\{ \int_{0}^{p}\left\vert \dot{W}_{p}(s)\right\vert 
\text{ }ds,\int_{0}^{p}\left\vert \dot{\nu}_{p}(s)\right\vert \text{ }%
ds,\int_{0}^{p}\left\vert \dot{N}_{p}(s)\right\vert \text{ }ds\right\}
\end{equation*}%
and remind that the manifold of substructural shapes is embedded in a linear
space. It then follows that%
\begin{equation*}
\left\vert w\left( K_{p},H\right) -w\left( K_{p},H^{\prime }\right)
\right\vert \leq \tilde{M}d\left( H,H^{\prime }\right)
\end{equation*}%
from which one realizes that $w\left( K_{p},\cdot \right) $ is a state
function. Previous inequality implies also that $w\left( K_{p},\cdot \right) 
$ is also Lipschitz continuous. The additivity of $w$\ with respect to the
processes is implied by the definition.
\end{proof}

The following lemma is a version of a proposition in \cite{DP} (see also 
\cite{Day3}), the proof of which can be easily adapted to the present case.

\begin{lemma}
Assume that $\nu $ can be freely selected in the linear space in which the
manifold of substructural shapes is isometrically embedded. Given two
different histories $H=\left( W,\nu ,N\right) $ and $H^{\prime }=\left(
W^{\prime },\nu ^{\prime },N^{\prime }\right) $ and a process $L_{p}\left(
r\right) $ defined by%
\begin{equation*}
L_{p}(r):=\left( \frac{p-r}{p}W(p)+\frac{r}{p}W^{\prime }(0),\frac{p-r}{p}%
\nu (p)+\frac{r}{p}\nu ^{\prime }(0),\frac{p-r}{p}N(p)+\frac{r}{p}N^{\prime
}(0)\right) ,
\end{equation*}%
one gets%
\begin{equation*}
\lim_{p\rightarrow +\infty }d(H_{p}\ast L_{p}\ast H^{\prime },H)=0
\end{equation*}%
and%
\begin{eqnarray*}
w(H_{p}\ast L_{p}\ast H^{\prime }) &=&w(H)+w(H^{\prime })+ \\
&&+\frac{1}{2}(\sum_{A,B}G_{AB}(\infty )B(\infty )\cdot B(\infty
)-\sum_{A,B^{\prime }}G_{AB^{\prime }}(\infty )B^{\prime }(0)\cdot B^{\prime
}(0)).
\end{eqnarray*}
\end{lemma}

\begin{definition}
A history $H$ is said to be $w-$approachable from another history $H^{\prime
}$ if $H$ is approachable from $H^{\prime }$ and the sequence $\left(
p\longmapsto K_{p},p\in \mathbb{R}^{+}\right) $ is such that the sequence $%
p\longmapsto w\left( K_{p},H\right) $ converges too.
\end{definition}

\begin{theorem}
If $\nu $ can be freely selected in the linear space in which the manifold
of substructural shapes is isometrically embedded, the space $\Gamma $ is
closed under $w-$approachability.
\end{theorem}

The proofs of both the previous lemma and the approachability theorem follow
the same paths of the analogous results for simple bodies in \cite{DP}. The
circumstance that $\nu $ is now selected in a linear space implies just
that, in re-following the path of the proofs in \cite{DP}, one needs only to
consider the distance $d\left( \cdot ,\cdot \right) $ and the presence of
the substructural terms.

In order to prove the closure theorem (under $w-$approachability) the key
point is the use of the last result of the Lemma. In fact, one replaces $%
H_{p}\ast L_{p}$ with 
\begin{equation*}
K_{2p}:=\left( H_{p}-H^{\prime }\left( 0\right) _{p}^{\dagger }\right) \ast
\left( L_{p}-H^{\prime }\left( 0\right) _{p}^{\dagger }\right) ,
\end{equation*}%
where $H^{\prime }\left( 0\right) _{p}^{\dagger }$ is the constant history
of value $H^{\prime }\left( 0\right) $ and duration $p$, then one proves by
Lemma that the work expended along the continuation $K_{2p}$, namely $%
w\left( K_{2p},H\right) $ converges to $w\left( H\right) $ plus the work
done in the extreme retardation (see \cite{DP}) from $H^{\prime }\left(
0\right) $ to $H\left( \infty \right) $.

The work $w$ helps also in characterizing the kernels in the constitutive
expressions of the interaction measures.

\begin{definition}
The relaxation functions $s\longmapsto G_{AB}\left( s\right) $ are said to
be dissipative if%
\begin{equation*}
w\left( H\right) \geq 0
\end{equation*}%
for any $H\in \Gamma $ such that $F\left( \infty \right) =0,$ $\nu \left(
\infty \right) =0,$ $N\left( \infty \right) =0$.
\end{definition}

Note that the requirement of the positivity of the work is weaker with
respect

\subsection{Relaxed work}

\begin{definition}
For every pair of $w-$approachable histories $H$ and $H^{\prime }$, the
relaxed work $w_{H^{\prime }}^{r}\left( H\right) $ along $H$, starting from $%
H^{\prime }$, is defined by%
\begin{equation*}
w_{H^{\prime }}^{r}\left( H\right) :=\inf \left\{ \liminf_{p\rightarrow
+\infty }w(K_{p},H^{\prime })\text{ }|\text{ }K_{p}\in \Pi
,\lim_{p\rightarrow +\infty }d(K_{p}\ast H^{\prime },H)=0\right\} .
\end{equation*}
\end{definition}

\begin{definition}
If $H$ is $w-$approachable from $H^{\prime }$, $w_{H^{\prime }}^{r}\left(
H\right) $ is called the \textbf{minimum work} performed from $H^{\prime }$
to $H$ while $-w_{H^{\prime }}^{r}\left( H\right) $ is called the \textbf{%
maximum recoverable work}.
\end{definition}

The closure of the state space under approachability justifies the
definition above. Moreover, an estimate follows:%
\begin{equation*}
w_{H^{\prime }}^{r}\left( H\right) \leq w\left( H\right) +\frac{1}{2}%
(\sum_{A,B}G_{AB}(\infty )B(\infty )\cdot B(\infty )-\sum_{A,B^{\prime
}}G_{AB^{\prime }}(\infty )B^{\prime }(0)\cdot B^{\prime }(0)).
\end{equation*}%
It means that the relaxed power along $H$, starting from $H^{\prime }$, is
bounded from above by the power along $H$ plus the difference of the powers
under extreme retardation from $H^{\prime }\left( 0\right) $ to $H\left(
\infty \right) $. By restricting $w_{H}^{r}\left( \cdot \right) $ to the
histories $K$ that are prolongations of $H$ itself, one gets an additional
upper bound:%
\begin{equation*}
w_{H}^{r}\left( K^{\dagger }\ast H\right) \leq w\left( K^{\dagger },H\right)
.
\end{equation*}

\begin{theorem}
The following statements hold:
\end{theorem}

\begin{enumerate}
\item \emph{Both }$w_{H^{\prime }}^{r}\left( \cdot \right) $ \emph{and }$%
w_{\left( \cdot \right) }^{r}\left( H\right) $ \emph{are state functions.}

\item (Sub-additivity.) \emph{For H, H' and H\textquotedblright\ histories
such that }$w_{H^{\prime }}^{r}\left( H\right) >-\infty $ \emph{and }$%
w_{H^{\prime \prime }}^{r}\left( H^{\prime }\right) >-\infty $ \emph{one
gets the triangular inequality}%
\begin{equation*}
w_{H^{\prime \prime }}^{r}\left( H\right) \leq w_{H^{\prime \prime
}}^{r}\left( H^{\prime }\right) +w_{H^{\prime }}^{r}\left( H\right) .
\end{equation*}

\item (Lower semicontinuity.) \emph{If }$w_{H^{\prime }}^{r}\left( H\right)
>-\infty $\emph{, }$w_{H^{\prime }}^{r}\left( \cdot \right) $ \emph{is lower
semicontinuous\footnote{%
Precisely, for any $\varepsilon >0$ there exists $\delta >0$ such that,
given $H$, for any $H_{1}$ such that $d\left( H_{1},H\right) <\delta $, one
gets 
\begin{equation*}
w_{H^{\prime }}^{r}\left( H_{1}\right) \geq w_{H^{\prime }}^{r}\left(
H\right) -\varepsilon .
\end{equation*}%
}.}

\item (Dissipation inequality.) \emph{If }$w_{H^{\prime }}^{r}\left(
H\right) >-\infty $\emph{\ and }$K_{p}$ \emph{is a process such that }$%
w_{H}^{r}\left( K_{p}\ast H\right) >-\infty $\emph{, then}%
\begin{equation*}
w_{H^{\prime }}^{r}\left( K_{d}\ast H\right) -w_{H^{\prime }}^{r}\left(
H\right) \leq w\left( K_{p}\ast H\right) ,
\end{equation*}%
\emph{moreover, if }$K_{p}$ \emph{is such that }$w_{H^{\prime }}^{r}\left(
K_{p}\ast H\right) >-\infty $ \emph{and }$w_{K_{p}\ast H^{\prime
}}^{r}\left( H\right) >-\infty $ \emph{then}%
\begin{equation*}
w_{H^{\prime }}^{r}\left( H\right) -w_{K_{p}\ast H^{\prime }}^{r}\left(
H\right) \leq w\left( K_{p},H^{\prime }\right) .
\end{equation*}
\end{enumerate}

The proof of the theorem above is essentially independent of the explicit
expression of the work. For this reason the proof of the analogous result in 
\cite{DP} applies providing one substitutes the norm used there with the
distance defined above. Other results in \cite{DP} can be adapted here. Such
results are listed below. Differences rest essentially on (i) the presence
of substructural terms, (ii) the use of the distance $d\left( \cdot ,\cdot
\right) $ and (iii) the use of strict continuations of histories, not the
relative continuations used in \cite{DP}. The latter are indicated by $R$
superposed to $\ast $ and are defined by%
\begin{equation*}
(K_{p}\overset{R}{\ast }H)(s):=\left\{ 
\begin{array}{ll}
K_{p}(s)+H\left( 0\right) & 0\leq s<d \\ 
H(s-d) & s\geq d%
\end{array}%
\right. .
\end{equation*}

\begin{theorem}
The following statements are equivalent:
\end{theorem}

\begin{enumerate}
\item $w$ \emph{satisfies the dissipation property on all constant histories.%
}

\item $w_{H^{\dagger }}^{r}\left( H^{\dagger }\right) =0$ \emph{along every
constant history }$H^{\dagger }$.

\item $w_{H^{\dagger }}^{r}\left( H\right) \geq 0$ \emph{for every history }$%
H$.

\item \emph{For any history }$H$\emph{\ one gets}%
\begin{equation*}
w(H)\geq \frac{1}{2}(\sum_{A,B}G_{AB}(\infty )B(0)\cdot B(0)-\sum_{{A,B}%
}G_{AB}(\infty )B(\infty )\cdot B(\infty )).
\end{equation*}

\item \emph{The relaxation functions} $s\longmapsto G_{AB}\left( s\right) $ 
\emph{are dissipative.}

\item $w_{H^{\prime }}^{r}\left( H\right) \geq 0$ \emph{for every pair of
histories }$H$ \emph{and }$H^{\prime }$\emph{.}

\item $w$\emph{\ satisfies the dissipation property for every history }$H$%
\emph{.}
\end{enumerate}

\begin{theorem}
The following statements hold:
\end{theorem}

\begin{enumerate}
\item \emph{For every history }$H$\emph{\ the minimum power performed and
the maximum recoverable work from }$H$\emph{\ to }$H$\emph{\ are zero.}

\item \emph{For every pair of histories }$H$\emph{\ and }$H^{\prime }$%
\begin{equation*}
w_{H^{\prime }}^{r}\left( H\right) \geq -w_{H}^{r}\left( H^{\prime }\right) .
\end{equation*}

\item \emph{For every constant history }$H^{\dagger }$\emph{, }$%
w_{H^{\dagger }}^{r}\left( \cdot \right) $\emph{\ and }$w_{\left( \cdot
\right) }^{r}\left( H^{\dagger }\right) $\emph{\ are determined respectively
by }$w_{0^{\dagger }}^{r}\left( \cdot \right) $\emph{\ and }$w_{\left( \cdot
\right) }^{r}\left( 0^{\dagger }\right) $\emph{, namely}%
\begin{equation*}
w_{H^{\dagger }}^{r}\left( H\right) =w_{0^{\dagger }}^{r}\left( H\right) -%
\frac{1}{2}\sum_{A,B}G_{AB}\left( \infty \right) B_{H^{\dagger }}\cdot
B_{H^{\dagger }},
\end{equation*}%
\begin{equation*}
w_{H}^{r}\left( H^{\dagger }\right) =w_{H}^{r}\left( 0^{\dagger }\right) +%
\frac{1}{2}\sum_{A,B}G_{AB}\left( \infty \right) B_{H^{\dagger }}\cdot
B_{H^{\dagger }}.
\end{equation*}

\item \emph{The restriction of }$w_{\left( \cdot \right) }^{r}\left( \cdot
\right) $\emph{\ to constant histories is determined: for every pair of
constant histories }$H_{1}^{\dagger }$\emph{\ and }$H_{2}^{\dagger }$ \emph{%
one gets}%
\begin{equation*}
w_{H_{1}^{\dagger }}^{r}\left( H_{2}^{\dagger }\right) =\frac{1}{2}%
(\sum_{A,B}G_{AB}\left( \infty \right) B_{H_{2}^{\dagger }}\cdot
B_{H_{2}^{\dagger }}-\sum_{A,B}G_{AB}\left( \infty \right) B_{H_{1}^{\dagger
}}\cdot B_{H_{1}^{\dagger }}).
\end{equation*}

\item \emph{For every pair of histories }$H$\emph{\ and }$H^{\prime }$%
\begin{equation*}
w_{H^{\prime }}^{r}\left( H\right) \geq w_{H^{\prime }}^{r}\left( H\left(
0\right) ^{\dagger }\right) =w_{H^{\prime }}^{r}(0^{\dag })+\frac{1}{2}%
\sum_{A,B}G_{AB}(\infty )B(0)\cdot B(0).
\end{equation*}

\item \emph{For every }$H$%
\begin{equation*}
-w_{H}^{r}(0^{\dag })\geq \frac{1}{2}\sum_{A,B}G_{AB}(\infty )B_{H(0)}\cdot
B_{H(0)}.
\end{equation*}

\item \emph{For every pair of histories }$H$\emph{\ and }$H^{\prime }$%
\begin{equation*}
w_{H^{\prime }}^{r}\left( H\right) \geq -w_{H^{\prime }}^{r}(0^{\dag })+%
\frac{1}{2}\sum_{A,B}G_{AB}(\infty )B_{H(0)}\cdot B_{H(0)}.
\end{equation*}

\item \emph{For every pair of histories }$H$\emph{\ and }$H^{\prime }$%
\begin{equation*}
w_{0^{\dag }}^{r}\left( H\right) \geq w_{H^{\prime }}^{r}\left( H\right) .
\end{equation*}

\item \emph{For every }$H$%
\begin{equation*}
w_{H^{\prime }}^{r}(0^{\dag })=\inf_{K\in \Pi }w(K,H).
\end{equation*}

\item \emph{For every constant history }$H^{\dagger }$\emph{, the functional 
}$w_{\left( \cdot \right) }^{r}\left( H^{\dagger }\right) $\emph{\ is upper
semicontinuous.}
\end{enumerate}

The symbols $B_{H^{\dagger }}$ and $B_{H\left( 0\right) }$ used earlier
indicate one of the entries of the list defining the state, evaluated along
the constant history $H^{\dagger }$ or at $H\left( 0\right) $, respectively.

\subsection{Free energies}

The free energy can be defined in terms of actions (see \cite{CO1}). There
are several possible free energies. Upper and lower bounds for their set can
be determined.

\begin{definition}
Every lower potential of $w$ is called a free energy.
\end{definition}

The properties of the power discussed above allow one to prove the following
theorem.

\begin{theorem}
The following assertions hold:
\end{theorem}

\begin{enumerate}
\item \emph{Every free energy }$\psi $\emph{\ satisfies the dissipation
inequality}%
\begin{equation*}
\psi \left( K\ast H\right) -\psi \left( H\right) <w\left( K,H\right)
\end{equation*}%
\emph{for every }$H\in \Gamma $\emph{\ and every compatible }$K\in \Pi $%
\emph{. Moreover, every l.s.c. function }$\psi :\Gamma \rightarrow \mathbb{R}
$\emph{\ that satisfies the dissipation inequality is a free energy.}

\item \emph{If the dissipation postulate is satisfied, then, for every }$H$%
\emph{\ and }$H^{\prime }$\emph{\ belonging to }$\Gamma $\emph{, the maps }$%
w_{H}^{r}\left( \cdot \right) $\emph{\ and }$-w_{\left( \cdot \right)
}^{r}\left( H^{\prime }\right) $\emph{\ are free energies. Moreover, for
every free energy }$\psi $\emph{\ one gets}%
\begin{equation*}
-w_{H}^{r}\left( H^{\prime }\right) \leq \psi \left( H\right) -\psi \left(
H^{\prime }\right) \leq w_{H}^{r}\left( H^{\prime }\right)
\end{equation*}%
\emph{for arbitrary histories }$H$\emph{\ and }$H^{\prime }$\emph{. In
particular, if there is }$H^{\prime }$\emph{\ and a family of free energies }%
$\psi _{s}$\emph{\ such that }$\psi _{s}\left( H^{\prime }\right) =0$\emph{,
the maps }$w_{H^{\prime }}^{r}\left( \cdot \right) $\emph{\ and }$-w_{\left(
\cdot \right) }^{r}\left( H^{\prime }\right) $\emph{\ are the maximum and
the minimum free energies in such a family.}

\item \emph{Each free energy is a state function and, for every }$H$\emph{,
it satisfies the inequality}%
\begin{equation*}
\psi \left( H\left( 0\right) ^{\dag }\right) \leq \psi \left( H\right) .
\end{equation*}%
\emph{In particular, the restriction of the free energy to constant
histories is given by}%
\begin{equation*}
\psi \left( H^{\dag }\right) -\psi \left( 0^{\dag }\right) =\frac{1}{2}%
\sum_{A,B}G_{AB}\left( \infty \right) B_{H^{\dag }}\cdot B_{H^{\dag }}.
\end{equation*}
\end{enumerate}

The technique of the proof is strictly analogous (modulo the variations
associated with the use of the metric $d\left( \cdot ,\cdot \right) $) to
the one used in \cite{DP} for an analogous result for simple viscoelastic
bodies (see also \cite{CO1}), so the details of the proof are not reported
here. In \cite{DP} a weaker condition is adopted: discontinuity is admitted
at $0$ under relative continuations. Here, the need of the use of the chain
rule in an ensuing section suggests to avoid this discontinuity for the sake
of simplicity.

\section{Characterization of the surface free energy in terms of the surface
work}

The propositions presented so far can be extended in presence of structured
discontinuity surfaces. In particular, it is assumed that across $\Sigma $
the map $\nu $ is continuous while the gradients $W$ and $N$ suffer bounded
jumps.

The attention is focused on constitutive relations of the type

\begin{equation*}
\mathbb{T}(\mathbb{H})=G_{\mathbb{TW}}(0)\mathbb{W}(0)+G_{\mathbb{T}\nu
}(0)\nu (0)+G_{\mathbb{TN}}(0)\mathbb{N}(0)+
\end{equation*}

\begin{equation*}
+\int_{0}^{+\infty }\left[ \dot{G}_{\mathbb{TW}}(s)\mathbb{W}(s)+\dot{G}_{%
\mathbb{T}\nu }(s)\nu (s)+\dot{G}_{\mathbb{TN}}(s)\mathbb{N}(s)\right] \text{
}ds,
\end{equation*}

\begin{equation*}
\mathbb{S}(\mathbb{H})=G_{\mathbb{SW}}(0)\mathbb{W}(0)+G_{\mathbb{S}\nu
}(0)\nu (0)+G_{\mathbb{SN}}(0)\mathbb{N}(0)+
\end{equation*}

\begin{equation*}
+\int_{0}^{+\infty }\left[ \dot{G}_{\mathbb{SW}}(s)\mathbb{W}(s)+\dot{G}_{%
\mathbb{S}\nu }(s)\nu (s)+\dot{G}_{\mathbb{SN}}(s)\mathbb{N}(s)\right] \text{
}ds,
\end{equation*}

\begin{equation*}
\mathfrak{z}(\mathbb{H})=G_{\mathfrak{z}\mathbb{W}}(0)\mathbb{W}(0)+G_{%
\mathfrak{z}\nu }(0)\nu (0)+G_{\mathfrak{z}\mathbb{n}}(0)\mathbb{W}(0)+
\end{equation*}

\begin{equation*}
+\int_{0}^{+\infty }\left[ \dot{G}_{\mathfrak{z}\mathbb{W}}(s)\mathbb{W}(s)+%
\dot{G}_{\mathfrak{z}\nu }(s)\nu (s)+\dot{G}_{\mathfrak{z}\mathbb{N}}(s)%
\mathbb{N}(s)\right] \text{ }ds,
\end{equation*}%
for the surface stresses $\mathbb{T}$, $\mathbb{S}$, and the surface
self-action $\mathfrak{z}$.

\begin{definition}
Two surface histories $\mathbb{H}$ and $\mathbb{H}^{\prime }$, such that $%
\mathbb{H}\left( 0\right) =\mathbb{H}^{\prime }\left( 0\right) $\ , are said
to be equivalent if%
\begin{equation*}
\mathbb{T}(\mathbb{K}_{p}\ast \mathbb{H})=\mathbb{T}(\mathbb{K}_{p}\ast 
\mathbb{H}^{\prime }),\text{ \ \ }\mathbb{S}(\mathbb{K}_{p}\ast \mathbb{H})=%
\mathbb{S}(\mathbb{K}_{p}\ast \mathbb{H}^{\prime }),\text{\ \ }\mathfrak{z}(%
\mathbb{K}_{p}\ast \mathbb{H})=\mathfrak{z}(\mathbb{K}_{p}\ast \mathbb{H}%
^{\prime }),
\end{equation*}%
for every prolongation $\mathbb{K}_{p}$.
\end{definition}

A distance $d_{\Sigma }(\mathbb{H},\mathbb{H}^{\prime })$ between surface
histories can be defined by

\begin{equation*}
d_{\Sigma }(\mathbb{H},\mathbb{H}^{\prime }):=\sup_{t>0}\left\{ \left\vert
\int_{0}^{+\infty }\left[ \dot{G}_{\mathbb{TW}}(s+t)\mathbb{W}(s)-\dot{G}_{%
\mathbb{TW}}(s+t)\mathbb{W}^{\prime }(s)+\dot{G}_{\mathbb{T}\nu }(s+t)\nu
(s)-\right. \right. \right.
\end{equation*}

\begin{equation*}
\left. \left. \left. -\dot{G}_{\mathbb{T}\nu }(s+t)(s)\nu ^{\prime }(s)+\dot{%
G}_{\mathbb{TN}}(s+t)\mathbb{N}(s)-\dot{G}_{\mathbb{TN}}(s+t)\mathbb{N}%
^{\prime }(s)\right] ds\right\vert +\right.
\end{equation*}

\begin{equation*}
\left. +\left\vert \int_{0}^{+\infty }\left[ \dot{G}_{\mathfrak{z}\mathbb{W}%
}(s+t)\mathbb{W}(s)-\dot{G}_{\mathfrak{z}\mathbb{W}}(s+t)\mathbb{W}^{\prime
}(s)+\dot{G}_{\mathfrak{z}\nu }(s+t)\nu (s)-\right. \right. \right.
\end{equation*}

\begin{equation*}
\left. \left. \left. -\dot{G}_{\mathfrak{z}\nu }(s+t)\nu ^{\prime }(s)+\dot{G%
}_{\mathfrak{z}\mathbb{N}}(s+t)\mathbb{N}(s)-\dot{G}_{\mathfrak{z}\mathbb{N}%
}(s+t)\mathbb{N}^{\prime }(s)\right] ds\right\vert +\right.
\end{equation*}

\begin{equation*}
\left. +\left\vert \int_{0}^{+\infty }\left[ \dot{G}_{\mathbb{SW}}(s+t)%
\mathbb{W}(s)-\dot{G}_{\mathbb{SW}}(s+t)\mathbb{W}^{\prime }(s)+\dot{G}_{%
\mathbb{S}\nu }(s+t)\nu (t)-\right. \right. \right.
\end{equation*}

\begin{equation*}
\left. \left. \left. -\dot{G}_{\mathbb{S}\nu }(s+t)\nu ^{\prime }(s)+\dot{G}%
_{\mathbb{SN}}(s+t)\mathbb{N}(s)-\dot{G}_{\mathbb{SN}}(s+t)\mathbb{N}%
^{\prime }(s)\right] ds\right\vert \right\} .
\end{equation*}

It is a semimetric on the space of surface histories and a metric over the
quotient space generated by the equivalence relation defined above.

\begin{proposition}
The distance $d_{\Sigma }\left( \cdot ,\cdot \right) $ has the following
properties:
\end{proposition}

\begin{description}
\item[Contraction] \emph{for every} $r>0$%
\begin{equation*}
d_{\Sigma }(\mathbb{K}_{r}\ast \mathbb{H}^{\prime },\mathbb{K}_{p}\ast 
\mathbb{H})\leq d_{\Sigma }(\mathbb{H}^{\prime },\mathbb{H})\quad \forall
r\geq 0.
\end{equation*}

\item[Fading memory] \emph{for} $\varepsilon >0$ \emph{there exists} $\ell $ 
\emph{such that, for every} $p>\ell $\emph{, one gets}%
\begin{equation*}
d_{\Sigma }(\mathbb{K}_{p}\ast \mathbb{H}^{\prime },\mathbb{K}_{p}\ast 
\mathbb{H})<\varepsilon .
\end{equation*}

\item[Approachability] \emph{if }$\mathbb{H}_{p}\left( p\right) ^{-}=\mathbb{%
H}^{\prime }\left( 0\right) $\emph{, then}%
\begin{equation*}
\lim_{p\rightarrow +\infty }d_{\Sigma }(\mathbb{H}_{p}\ast \mathbb{H}%
^{\prime },\mathbb{H})=0,
\end{equation*}%
\emph{with} $\mathbb{H}_{p}$ \emph{the process generated by} $\mathbb{H}$ 
\emph{over} $[0,p)$.
\end{description}

The proof is analogous to the one of Proposition 1. In this case the
constant $M$ is the maximum of the suprema of the surface histories.

Surface state functions and surface actions can be then defined.

From Theorem 1 one realizes that the \emph{surface work density} $w^{\Sigma
} $\ is defined by%
\begin{equation*}
w^{\Sigma }:=\int_{0}^{+\infty }\left( \mathbb{T}\cdot \dot{\mathbb{W}}+%
\mathfrak{z}\cdot \left\langle \dot{\nu}\right\rangle +\mathbb{S}\cdot \dot{%
\mathbb{N}}\right) \text{ }dt+\int_{0}^{+\infty }\left( \langle \sigma
\rangle m\cdot \left[ \dot{y}\right] +\langle \mathcal{S}\rangle m\cdot %
\left[ \dot{\nu}\right] \right) \text{ }dt.
\end{equation*}%
It includes both peculiar surface interactions and traces of the bulk
stresses at the discontinuity surface itself. Previous work on the
instantaneous response of complex bodies with structured discontinuity
surfaces \cite{M02} suggests that only a \emph{reduced surface work density} 
$\hat{w}^{\Sigma }$ is in strict connection with the surface energy density:%
\begin{equation*}
\hat{w}^{\Sigma }:=\int_{0}^{+\infty }\left( \mathbb{T}\cdot \dot{\mathbb{W}}%
+\mathfrak{z}\cdot \left\langle \dot{\nu}\right\rangle +\mathbb{S}\cdot \dot{%
\mathbb{N}}\right) \text{ }dt.
\end{equation*}

The reduced surface work density over prolongations is defined by%
\begin{equation*}
\hat{w}^{\Sigma }\left( \mathbb{K}_{p},\mathbb{H}\right) :=\hat{w}^{\Sigma
}\left( \mathbb{K}_{p}\ast \mathbb{H}\right) -\hat{w}^{\Sigma }\left( 
\mathbb{H}\right) ,
\end{equation*}%
where $\mathbb{K}_{p}$\ is the surface counterpart of $K_{p}$.

By making use of the technique leading to Theorem 2, one may prove that $%
\hat{w}^{\Sigma }$ is an action and the map $\mathbb{H}\longmapsto \hat{w}%
^{\Sigma }\left( \mathbb{K}_{p},\mathbb{H}\right) $ is a state function.

A \emph{relaxed} \emph{surface work} $\hat{w}_{\mathbb{H}^{\prime
}}^{r\Sigma }$ can be then defined by%
\begin{equation*}
\hat{w}_{\mathbb{H}^{\prime }}^{r\Sigma }\left( \mathbb{H}\right) :=\inf
\left\{ \liminf_{p\rightarrow +\infty }\hat{w}^{\Sigma }\left( \mathbb{K}%
_{p},\mathbb{H}^{\prime }\right) \text{ }|\text{ }\mathbb{K}_{p}\in \Pi
,\lim_{p\rightarrow +\infty }d(\mathbb{K}_{p}\ast \mathbb{H}^{\prime },%
\mathbb{H})=0\right\} .
\end{equation*}

As in the case of the bulk relaxed work, both $\hat{w}_{\mathbb{H}^{\prime
}}^{r\Sigma }\left( \mathbb{\cdot }\right) $ and $\hat{w}_{\left( \cdot
\right) }^{r\Sigma }\left( \mathbb{H}\right) $ are state functions.
Moreover, $\hat{w}_{\left( \cdot \right) }^{r\Sigma }\left( \mathbb{\cdot }%
\right) $ is subadditive in the sense that, for histories $\mathbb{H}$, $%
\mathbb{H}^{\prime }$ and $\mathbb{H}^{\prime \prime }$ such that $\hat{w}_{%
\mathbb{H}^{\prime }}^{r\Sigma }\left( \mathbb{H}\right) >-\infty $ and $%
\hat{w}_{\mathbb{H}^{\prime \prime }}^{r\Sigma }\left( \mathbb{H}^{\prime
}\right) >-\infty $, the inequality%
\begin{equation*}
\hat{w}_{\mathbb{H}^{\prime \prime }}^{r\Sigma }\left( \mathbb{H}\right)
\leq \hat{w}_{\mathbb{H}^{\prime \prime }}^{r\Sigma }\left( \mathbb{H}%
^{\prime }\right) +\hat{w}_{\mathbb{H}^{\prime }}^{r\Sigma }\left( \mathbb{H}%
\right)
\end{equation*}%
holds. If $\hat{w}_{\mathbb{H}^{\prime }}^{r\Sigma }\left( \mathbb{H}\right)
>-\infty $, $\hat{w}_{\mathbb{H}^{\prime }}^{r\Sigma }\left( \mathbb{\cdot }%
\right) $ is lower semicontinuous. Under the same hypothesis, for every
process $\mathbb{K}_{p}$ such that $\hat{w}_{\mathbb{H}}^{r\Sigma }\left( 
\mathbb{K}_{p}\ast \mathbb{H}\right) >-\infty $, the surface dissipation
inequality%
\begin{equation*}
\hat{w}_{\mathbb{H}^{\prime }}^{r\Sigma }\left( \mathbb{K}_{p}\ast \mathbb{H}%
\right) -\hat{w}_{\mathbb{H}^{\prime }}^{r\Sigma }\left( \mathbb{H}\right)
\leq w^{\Sigma }\left( \mathbb{K}_{p}\ast \mathbb{H}\right)
\end{equation*}%
is verified.

In summary, all the properties of $w$ and $w^{r}$ hold also for $\hat{w}%
^{\Sigma }$ and $\hat{w}^{r\Sigma }$. The proofs can be constructed in the
same way adopted in analyzing the power in the bulk.

\begin{definition}
Every lower potential of $w^{\Sigma }$ is called a surface free energy.
\end{definition}

The theorem collecting the properties of the free energy in the bulk has its
counterpart for the surface free energy. The proof is essentially the same.

\begin{theorem}
The following assertions hold:
\end{theorem}

\begin{enumerate}
\item \emph{Every free energy }$\phi $\emph{\ satisfies the dissipation
inequality}%
\begin{equation*}
\phi \left( \mathbb{K}\ast \mathbb{H}\right) -\phi \left( \mathbb{H}\right)
<w_{\Sigma }\left( \mathbb{K},\mathbb{H}\right)
\end{equation*}%
\emph{for every }$\mathbb{H}$\emph{\ and every compatible }$\mathbb{K}$\emph{%
. Moreover, every l.s.c. function }$\mathbb{H}\longmapsto \phi \left( 
\mathbb{H}\right) \in \mathbb{R}$\emph{\ that satisfies the dissipation
inequality is a free energy.}

\item \emph{If the dissipation postulate is satisfied, then, for every pair
of histories }$\mathbb{H}$\emph{\ and }$\mathbb{H}^{\prime }$\emph{, the
maps }$\hat{w}_{\mathbb{H}}^{r\Sigma }\left( \cdot \right) $\emph{\ and }$-%
\hat{w}_{\left( \cdot \right) }^{r\Sigma }\left( \mathbb{H}^{\prime }\right) 
$\emph{\ are free energies. Moreover, for every free energy }$\phi $\emph{\
one gets}%
\begin{equation*}
-\hat{w}_{\mathbb{H}}^{r\Sigma }\left( \mathbb{H}^{\prime }\right) \leq \phi
\left( \mathbb{H}\right) -\phi \left( \mathbb{H}^{\prime }\right) \leq \hat{w%
}_{\mathbb{H}}^{r\Sigma }\left( \mathbb{H}^{\prime }\right)
\end{equation*}%
\emph{for arbitrary histories }$\mathbb{H}$\emph{\ and }$\mathbb{H}^{\prime
} $\emph{. In particular, if there is }$\mathbb{H}^{\prime }$\emph{\ and a
family of free energies }$\phi _{s}$\emph{\ such that }$\phi _{s}\left( 
\mathbb{H}^{\prime }\right) =0$\emph{, the maps }$\hat{w}_{\mathbb{H}%
^{\prime }}^{r\Sigma }\left( \cdot \right) $\emph{\ and }$-\hat{w}_{\left( 
\mathbb{\cdot }\right) }^{r\Sigma }\left( \mathbb{H}^{\prime }\right) $\emph{%
\ are the maximum and the minimum free energies in such a family.}

\item \emph{Every free energy is a state function and, for every }$\mathbb{H}
$\emph{, it satisfies the inequality}%
\begin{equation*}
\phi \left( \mathbb{H}\left( 0\right) ^{\dag }\right) \leq \phi \left( 
\mathbb{H}\right) .
\end{equation*}%
\emph{In particular, the restriction of the free energy to constant
histories is given by}%
\begin{equation*}
\phi \left( \mathbb{H}^{\dag }\right) -\phi \left( 0^{\dag }\right) =\frac{1%
}{2}\sum_{\mathbb{A},\mathbb{B}}G_{\mathbb{AB}}\left( \infty \right) \mathbb{%
B}_{\mathbb{H}^{\dag }}\cdot \mathbb{B}_{\mathbb{H}^{\dag }}.
\end{equation*}
\end{enumerate}

$\mathbb{A}$ ranges in $\left\{ \mathbb{T},\mathfrak{z},\mathbb{S}\right\} $
and $\mathbb{B}$ in $\left\{ \mathbb{W},\nu ,\mathbb{N}\right\} $.

\section{The mechanical dissipation inequality and its consequences}

\subsection{Mechanical dissipation inequality}

In isothermal setting, the second law of thermodynamics in the form of
Clausius-Duhem inequality reduces to a mechanical dissipation inequality. In
Lagrangian representation, for any part $\mathfrak{b}$ it reads%
\begin{equation*}
\frac{d}{dt}\Psi \left( \mathfrak{b},y,\nu \right) -\mathcal{P}_{\mathfrak{b}%
}^{ext}\left( \dot{y},\dot{\nu}\right) \leq 0.
\end{equation*}%
The functional $\Psi $ is the overall free energy of $\mathfrak{b}$ along
the fields $y$ and $\nu $: it is the integral over $\mathfrak{b}$\ itself of
the free energy density. If $\mathfrak{b}$ is selected to cross $\Sigma $,
so it is indicated by $\mathfrak{b}_{\Sigma }$, both bulk and surface energy
densities - the ones discussed previously - must be involved. Local forms of
the mechanical dissipation inequality - local in the bulk and along $\Sigma $
- can be obtained by exploiting the arbitrariness of the part considered.
They are reported in summary here, written with reference to the
infinitesimal deformation setting discussed in the earlier sections. The
local form of the mechanical dissipation inequality in the bulk then reads%
\begin{equation*}
\dot{\psi}-\sigma \cdot \dot{W}-z\cdot \dot{\nu}-\mathcal{S}\cdot \dot{N}%
\leq 0,
\end{equation*}%
while the one at points over $\Sigma $ \ is given by%
\begin{equation*}
\dot{\phi}-\mathbb{T\cdot \dot{W}-\mathfrak{z}\cdot }\dot{\nu}-\mathbb{%
S\cdot \dot{N}}\leq 0.
\end{equation*}

\subsection{Complex bodies with instantaneous elastic response}

By borrowing terms from the mechanics of simple bodies, here bodies with
instantaneous elastic response are the ones admitting constitutive
structures of the type%
\begin{equation*}
\psi =\psi \left( H\left( t\right) ,H^{t}\right) ,
\end{equation*}%
\begin{equation*}
\sigma \mathbb{=\sigma }\left( H\left( t\right) ,H^{t}\right) ,
\end{equation*}%
\begin{equation*}
z=z\left( H\left( t\right) ,H^{t}\right) ,
\end{equation*}%
\begin{equation*}
\mathcal{S}=\mathcal{S}\left( H\left( t\right) ,H^{t}\right) ,
\end{equation*}%
in the bulk and 
\begin{equation*}
\phi =\phi \left( \mathbb{H}\left( t\right) ,\mathbb{H}^{t}\right) ,
\end{equation*}%
\begin{equation*}
\mathbb{T}=\mathbb{T}\left( \mathbb{H}\left( t\right) ,\mathbb{H}^{t}\right)
,
\end{equation*}%
\begin{equation*}
\mathfrak{z}=\mathfrak{z}\left( \mathbb{H}\left( t\right) ,\mathbb{H}%
^{t}\right) ,
\end{equation*}%
\begin{equation*}
\mathbb{S}=\mathbb{S}\left( \mathbb{H}\left( t\right) ,\mathbb{H}^{t}\right)
.
\end{equation*}%
The symbols $H$ and $\mathbb{H}$\ summarize the state. Precisely, $H\left(
t\right) =\left( W\left( t\right) ,\nu \left( t\right) ,N\left( t\right)
\right) \in M_{3\times 3}\times \mathbb{R}^{k}\times M_{k\times 3}$ is the
state at the instant $t$ while $H^{t}\in \Gamma $ is the past history of the
state up to the instant $t$ (in the notation used here $H^{t}$ is the graph
of $H$ from $t$ to infinity, namely $H^{t}\left( s\right) =H\left(
t+s\right) $). Analogous meaning can be attributed to $\mathbb{H}\left(
t\right) $ and $\mathbb{H}^{t}$, namely $\mathbb{H}\left( t\right) =\left( 
\mathbb{W}\left( t\right) ,\nu \left( t\right) ,\mathbb{N}\left( t\right)
\right) \in M_{3\times 3}\times \mathbb{R}^{k}\times M_{k\times 3}$ while $%
\mathbb{H}^{t}$ is the past surface history and belongs to $\Gamma $.

In the infinitesimal deformation setting treated here, it is assumed that
the \textbf{free energy density in the bulk} is a \textbf{quadratic form} in
the instantaneous values $W\left( t\right) $, $\nu \left( t\right) $ and $%
N\left( t\right) $. It is also assumed that also the \textbf{surface free
energy density} is a \textbf{quadratic form} in the instantaneous values $%
\mathbb{W}\left( t\right) $, $\nu \left( t\right) $ and $\mathbb{N}\left(
t\right) $.

\subsection{Chain rule}

To exploit the local versions of the mechanical dissipation inequality a
chain rule must be used in evaluating the time derivative of the energy.
Appropriate chain rules have been obtained in \cite{MiWa} and \cite{Day3},
and can be adapted here.

Consider a functional%
\begin{equation*}
\mathcal{F}:M_{3\times 3}\times \mathbb{R}^{k}\times M_{k\times 3}\times
\Gamma \rightarrow \mathbb{R}
\end{equation*}%
defined for every $H\left( t\right) $ (or $\mathbb{H}\left( t\right) $) in $%
M_{3\times 3}\times \mathbb{R}^{k}\times M_{k\times 3}$ and for every $H^{t}$
(or $\mathbb{H}^{t}$) in $\Gamma $ such that $H^{t}\left( s\right) $\ (or $%
\mathbb{H}^{t}\left( s\right) $) is in the open and connected subset $%
\mathcal{U}$ from $M_{3\times 3}\times \mathbb{R}^{k}\times M_{k\times 3}$,
characterized by $\det \left( I+W\right) >0$ (or $\det \left( I+\mathbb{W}%
\right) >0$), for almost $s>0$ ($s$ is the time parametrizing the history
`prior' $t$ - in the representation adopted here $s>t$). Assume that (\emph{i%
}) $\mathcal{F}$ \ is continuously differentiable, (\emph{ii}) the function $%
t\longmapsto H\left( t\right) $ (or $\mathbb{H}\left( t\right) $) with
values in $\mathcal{U}$ has two continuous derivatives $t\longmapsto \dot{H}%
\left( t\right) $ and $t\longmapsto \ddot{H}\left( t\right) $, and (\emph{iii%
}) for every $t$ the past histories $\dot{H}^{t}$ and $\ddot{H}^{t}$ are in $%
\Gamma $. Under these conditions the function $f\left( t\right) :=\mathcal{F}%
\left( H\left( t\right) ,H^{t}\right) $ (alternatively $f\left( t\right) :=%
\mathcal{F}\left( \mathbb{H}\left( t\right) ,\mathbb{H}^{t}\right) $) is
continuously differentiable and its time derivative is%
\begin{equation*}
\dot{f}\left( t\right) =D\mathcal{F}\left( H\left( t\right) ,H^{t}\right)
\cdot \dot{H}\left( t\right) +\delta \mathcal{F}\left( H\left( t\right)
,H^{t}|\dot{H}^{t}\right)
\end{equation*}%
(alternatively $\dot{f}\left( t\right) =D\mathcal{F}\left( \mathbb{H}\left(
t\right) ,\mathbb{H}^{t}\right) \cdot \mathbb{\dot{H}}\left( t\right)
+\delta \mathcal{F(}\mathbb{H}\left( t\right) ,\mathbb{H}^{t}|\mathbb{\dot{H}%
}^{t})$), where $D\mathcal{F}\left( H\left( t\right) ,H^{t}\right) $ is a
continuous functional taking values in $T_{H\left( t\right) }^{\ast }%
\mathcal{U}$ for every fixed $H\left( t\right) $ and $H^{t}$, and $\delta 
\mathcal{F}\left( H\left( t\right) ,H^{t}|K^{t}\right) $ is a continuous
scalar-valued functional depending linearly on $K^{t}$ and defined on the
closed subspace of $\Gamma $ spanned by the functions $K^{t}$ such that $%
H^{t}\left( s\right) +K^{t}\left( s\right) $ is in $\mathcal{U}$ for almost $%
s>0$ (analogous remarks hold also for $D\mathcal{F}\left( \mathbb{H}\left(
t\right) ,\mathbb{H}^{t}\right) $ and $\delta \mathcal{F(}\mathbb{H}\left(
t\right) ,\mathbb{H}^{t}|\mathbb{\dot{H}}^{t})$). Fixed $K^{t}$, an
appropriate technical assumption is that $\delta \mathcal{F}\left( H\left(
t\right) ,H^{t}|K^{t}\right) $ is continuous in $\left( H\left( t\right)
,H^{t}\right) $. The proof of this chain rule can be found in \cite{Day3}.

\subsection{Consequences of the mechanical dissipation inequality}

As it is well known, to exploit the local version of the mechanical
dissipation inequality one should have the possibility to select at will the
instantaneous rate $\dot{H}\left( t\right) $ of the state. If this point is
straightforward in the mechanics of (simple or complex) bodies without
memory effects, some additional problems appear in presence of memory
effects, due to the dependence of the interaction measures on the whole
history of the state variables. The technique discussed in \cite{Day3} to
avoid these difficulties can be adapted here and is summarized in the
following paragraphs, then it is applied to the case of complex bodies.

For instrumental reasons, it is useful to introduce a $C^{2}$ function $f:%
\mathbb{R}^{+}\mathbb{\rightarrow R}$ such that $f\left( s\right) =0$ for $%
\left\vert s\right\vert \geq 1$, $f\left( 0\right) =0$, $\dot{f}\left(
0\right) =1$. By using $f$, for any $\alpha \in \mathbb{R}^{+}$\ and a fixed 
$t\in \mathbb{R}^{+}$ one may define (see [Day], p. 91) \emph{varied}
histories%
\begin{equation*}
H_{\alpha }\left( s\right) :=H\left( s\right) +\alpha f\left( \frac{s-t}{%
\alpha }\right) \left( \mathcal{H}-\dot{H}\left( t\right) \right) ,
\end{equation*}%
where $\mathcal{H}$ is a generic element from $T_{H\left( t\right) }\mathcal{%
U}$, namely $\mathcal{H}$ is a triple $\left( \mathcal{V},\upsilon ,\Upsilon
\right) $ of virtual rates of $W$, $\nu $ and $N$. Essential properties of $%
H_{\alpha }\left( t\right) $ - in a sense the properties that suggest the
definition of $H_{\alpha }\left( t\right) $ itself - are (\emph{i}) $%
H_{\alpha }\left( t\right) =H\left( t\right) $, (\emph{ii}) $\dot{H}_{\alpha
}\left( t\right) =\mathcal{H}$, (\emph{iii}) $H_{\alpha }\left( s\right)
=H\left( s\right) $ for every $s\leq t-\alpha $ and $s\geq t+\alpha $ and (%
\emph{iv}) for $\alpha $ sufficiently small $H_{\alpha }\left( \cdot \right) 
$ meets the hypotheses of the chain rule and both $H_{\alpha }\left( \cdot
\right) $ and $\dot{H}_{\alpha }\left( \cdot \right) $ converge in norm
respectively to $H\left( \cdot \right) $ and $\dot{H}\left( \cdot \right) $
as $\alpha \rightarrow 0$.

A similar definition can be adopted for the surface history $\mathbb{H}%
_{\alpha }\left( s\right) $ so that one gets%
\begin{equation*}
\mathbb{H}_{\alpha }\left( s\right) :=\mathbb{H}\left( s\right) +\alpha
f\left( \frac{s-t}{\alpha }\right) \left( \mathfrak{H}-\mathbb{\dot{H}}%
\left( t\right) \right) ,
\end{equation*}%
where, now, $\mathfrak{H}$ is a generic element from $T_{\mathbb{H}\left(
t\right) }\mathcal{U}$, namely $\mathfrak{H}$\ is a triple $\left( \mathfrak{%
V},\upsilon ,\mathfrak{G}\right) $ of virtual rates of $\mathbb{W}$, $\nu $
and $\mathbb{N}$.

By making use of the chain rule and substituting $H\left( s\right) $ with $%
H_{\alpha }\left( s\right) $, from the local mechanical dissipation
inequality in the bulk one gets%
\begin{equation*}
(\partial _{W\left( t\right) }\psi \left( H\left( t\right) ,H^{t}\right)
-\sigma \left( H\left( t\right) ,H^{t}\right) )\cdot \mathcal{V}+(\partial
_{\nu \left( t\right) }\psi \left( H\left( t\right) ,H^{t}\right) -z\left(
H\left( t\right) ,H^{t}\right) )\cdot \upsilon +
\end{equation*}%
\begin{equation*}
+(\partial _{N\left( t\right) }\psi \left( H\left( t\right) ,H^{t}\right) -%
\mathcal{S}\left( H\left( t\right) ,H^{t}\right) )\cdot \Upsilon +\delta
\psi \left( H\left( t\right) ,H^{t}|\dot{H}^{t}\right) \leq 0
\end{equation*}%
as $\alpha \rightarrow 0$, an inequality holding for all choices of the
triple $\left( \mathcal{V},\upsilon ,\Upsilon \right) $. The arbitrariness
of $\left( \mathcal{V},\upsilon ,\Upsilon \right) $ implies that in the bulk%
\begin{equation*}
\sigma \left( H\left( t\right) ,H^{t}\right) =\partial _{W\left( t\right)
}\psi \left( H\left( t\right) ,H^{t}\right) ,
\end{equation*}%
\begin{equation*}
z\left( H\left( t\right) ,H^{t}\right) =\partial _{\nu \left( t\right) }\psi
\left( H\left( t\right) ,H^{t}\right) ,
\end{equation*}%
\begin{equation*}
\mathcal{S}\left( H\left( t\right) ,H^{t}\right) =\partial _{N\left(
t\right) }\psi \left( H\left( t\right) ,H^{t}\right) ,
\end{equation*}%
\begin{equation*}
\delta \psi \left( H\left( t\right) ,H^{t}|\dot{H}^{t}\right) \leq 0.
\end{equation*}

An analogous result hold along the surface $\Sigma $ where, locally one gets%
\begin{equation*}
(\partial _{\mathbb{W}\left( t\right) }\phi \left( \mathbb{H}\left( t\right)
,\mathbb{H}^{t}\right) -\mathbb{T}\left( \mathbb{H}\left( t\right) ,\mathbb{H%
}^{t}\right) )\cdot \mathfrak{V}+(\partial _{\nu \left( t\right) }\phi
\left( \mathbb{H}\left( t\right) ,\mathbb{H}^{t}\right) -\mathfrak{z}\left( 
\mathbb{H}\left( t\right) ,\mathbb{H}^{t}\right) )\cdot \upsilon +
\end{equation*}%
\begin{equation*}
+(\partial _{\mathbb{N}\left( t\right) }\phi \left( \mathbb{H}\left(
t\right) ,\mathbb{H}^{t}\right) -\mathcal{S}\left( \mathbb{H}\left( t\right)
,\mathbb{H}^{t}\right) )\cdot \mathfrak{G}+\delta \phi \left( \mathbb{H}%
\left( t\right) ,\mathbb{H}^{t}|\mathbb{\dot{H}}^{t}\right) \leq 0
\end{equation*}%
as $\alpha \rightarrow 0$, an inequality holding for all choices of the
triple $\left( \mathfrak{V},\upsilon ,\mathfrak{G}\right) $. The
arbitrariness of $\left( \mathfrak{V},\upsilon ,\mathfrak{G}\right) $
implies that, along the surface,%
\begin{equation*}
\mathbb{T}\left( \mathbb{H}\left( t\right) ,\mathbb{H}^{t}\right) =\partial
_{\mathbb{W}\left( t\right) }\phi \left( \mathbb{H}\left( t\right) ,\mathbb{H%
}^{t}\right) ,
\end{equation*}%
\begin{equation*}
\mathfrak{z}\left( \mathbb{H}\left( t\right) ,\mathbb{H}^{t}\right)
=\partial _{\nu \left( t\right) }\phi \left( \mathbb{H}\left( t\right) ,%
\mathbb{H}^{t}\right) ,
\end{equation*}%
\begin{equation*}
\mathcal{S}\left( \mathbb{H}\left( t\right) ,\mathbb{H}^{t}\right) =\partial
_{\mathbb{N}\left( t\right) }\phi \left( \mathbb{H}\left( t\right) ,\mathbb{H%
}^{t}\right) ,
\end{equation*}%
\begin{equation*}
\delta \phi \left( \mathbb{H}\left( t\right) ,\mathbb{H}^{t}|\mathbb{\dot{H}}%
^{t}\right) \leq 0.
\end{equation*}

\textbf{Acknowledgements}. The support of the GNFM-CNR is acknowledged. This
work has been developed within the programs of the research group in
"Theoretical Mechanics" of the "Centro di Ricerca Matematica Ennio De
Giorgi" of the Scuola Normale Superiore at Pisa.

\end{document}